\documentclass[aps,prx,twocolumn,nofootinbib,superscriptaddress,longbibliography]{revtex4-2}

\usepackage{amsmath,amssymb,amsfonts,amsthm,mathtools}
\usepackage{graphicx}
\usepackage{booktabs}
\usepackage[colorlinks=true,citecolor=blue,linkcolor=blue,urlcolor=blue]{hyperref}
\usepackage{physics}
\usepackage{bm}
\newcommand{\cH}{\mathcal H}
\newcommand{\cE}{\mathcal E}

\newcommand{\Span}{\operatorname{span}}

\newcommand{\Part}{\operatorname{Part}}
\newcommand{\eqp}{\operatorname{eq}}

\usepackage{color}

\newcommand{\calE}{{\cal E}}

\newcommand{\calM}{{\cal M}}

\newtheorem{theorem}{Theorem}

\newtheorem{lemma}{Lemma}
\newtheorem{corollary}{Corollary}

\begin{document}

\title{Hypothesis testing between quantum ensembles}

\author{Jian Yao}
\affiliation{Ming Hsieh Department of Electrical and Computer Engineering, University of Southern California, Los Angeles, California 90089, USA}

\author{Quntao Zhuang}
\email{qzhuang@usc.edu}
\affiliation{Ming Hsieh Department of Electrical and Computer Engineering, University of Southern California, Los Angeles, California 90089, USA}
\affiliation{Department of Physics and Astronomy, University of Southern California, Los Angeles, California 90089, USA}


\begin{abstract}
Quantum state ensembles are important in quantum information processing. For example, quantum $t$-designs model highly entangled states in complex systems, while projected ensembles appear in generative quantum machine learning and studies of thermalization. With their sample state accompanied by a classical label, these ensembles contain operational information beyond their average density operators. Yet an ensemble differs from a classical-quantum state because it is invariant under permutations of labels.
We formulate binary hypothesis testing between finite quantum ensembles and derive fundamental limits on error probability. Given an observed label pattern, we show that the joint sampled state can be described by power-weighted ensemble moments. This yields the Bayes-optimal measurement and exact finite-sample error, revealing that discrimination is governed by the full moment hierarchy up to the number of samples. In the many-sample limit, we derive Chernoff bounds and obtain exact error exponents for finite uniform pure-state ensembles. We apply these results to optical communication and $t$-designs. For finite uniform pure-state $t$-designs with large $t$, the maximal discrimination exponent scales sharply as $\sim t^{-2}$, while equal-prior fixed-error testing requires $\sim t^2$ samples.


\end{abstract}

\maketitle

Ensembles of quantum states are central objects in quantum information processing. While a single quantum state represents the description of an individual quantum system, an ensemble of states provides a statistical or Bayesian description when the system parameters are only partially specified. For example, in the study of chaotic many-body quantum systems, one often adopts a typical state drawn from the Haar-random ensemble or a $t$-design~\cite{ambainis2007quantum,roberts2017chaos}, subject to conservation laws~\cite{goldstein2006canonical,murciano2022symmetry}, and uses this ensemble description to derive universal properties of entanglement, purity, and spectral statistics~\cite{lubkin1978entropy,page1993average,foong1994proof,popescu2006entanglement,goldstein2006canonical}. Moreover, projected ensembles---ensembles of states generated by measuring a subsystem are crucial in the study of deep thermalization~\cite{ho2022exact,ippoliti2022solvable,cotler2023emergent,ippoliti2023dynamical,choi2023preparing,zhang2025holographic} and quantum generative models~\cite{zhang2024generative,kwun2025mixed, mo2026measurementinducedoverconcentrationquantumgenerative}.
In optical communication, the encoding quantum signal constellation is also specified by a quantum ensemble~\cite{helstrom1976quantum,holevo2012quantum}.

In these scenarios, it is important to distinguish an ensemble of quantum states from its average quantum state. For instance, the average of Haar-random pure states is a trivial maximally mixed state. 
To go beyond the average state, one must adopt a classical label for each state from the ensemble. Indeed, in a complex quantum dynamics, such labels may be the Hamiltonian parameters, circuit parameters or random seed. In projected and measurement-induced ensemble, the classical label is the measurement outcome.
However, a quantum state ensemble is also different from a classical-quantum state in which the precise label-state correspondence is part of the physical description. For example, a $t$-design remains the same $t$-design after an arbitrary shuffling of the labels. Thus, an ensemble is an object with more structure than its average state, but less label structure than a classical-quantum state. 

Throughout this work, we adopt this label-invariant viewpoint: a finite quantum ensemble is an unordered collection of weighted quantum states, and two descriptions that differ only by relabeling represent the same ensemble.
This label-invariant viewpoint also connects to classical distribution testing and statistical learning. When data are relabeled, the data set itself remains unchanged, and label-invariant quantities are known as symmetric properties~\cite{valiant2011power,valiant2017estimating,valiant2008testing,pmlr-v70-acharya17a,NEURIPS2019_f19fec2f}.

Although quantum ensembles are widely used throughout quantum information theory, the fundamental theory of estimating, learning, and distinguishing them remains less developed. In this work, we study hypothesis testing between quantum ensembles under label-invariant access and establish fundamental limits on the resources required to distinguish them. We show that the state description of an ensemble reduces to moments, derive the minimum Bayes error, and provide bounds on the asymptotic error exponent. We provide analytical examples, including the hypothesis testing between state designs. In particular, for fixed $d$, the maximal Chernoff exponent over pairs of finite uniform pure-state $t$-designs scales sharply as $\Theta_d(t^{-2})$, while under equal priors, reducing the Bayes error below any fixed $\epsilon\in(0,1/2)$ requires $\Omega_{d,\epsilon}(t^2)$ samples as $t\to\infty$. 

 {\em Quantum ensembles and the state description.---}
A finite quantum state ensemble is an unordered collection \(\cE=\{(p_i,\rho_i)\}_{i=1}^N\), where \(p_i>0\), \(\sum_{i=1}^Np_i=1\), and the $\rho_i$ are pairwise distinct density operators on a Hilbert space $\cH$, i.e., $\rho_i\neq\rho_j$ for all $i\neq j$. Since the component labels have no intrinsic meaning, descriptions related by a permutation represent the same ensemble, i.e., \(\{(p_i,\rho_i)\}_{i=1}^N=\{(p_{\pi(i)},\rho_{\pi(i)})\}_{i=1}^N\) for every \(\pi\in S_N\).

 Below, we establish the quantum state description of the sampling process. Since a sample without any label is described by the average state $\sum_i p_i \rho_i$, we assume a single sampling process returns a state $\rho_l$ and its label $\lambda_l$, while the intrinsic label $\ell$ specifying the state is unknown. 
 Repeated occurrences of the same label identify repeated draws of the same state, but the correspondence between visible labels and the intrinsic state indices is unknown and carries no physical significance---only the equality pattern of the observed labels is informative.
 
 Consider an $N=2$ ensemble as an example, when one obtains a sample $(1,\sigma_1)$, it is unclear if one obtained $\sigma_1=\rho_1$ or $\sigma_1=\rho_2$. And when one obtains physical samples \((1,\sigma_1),(1,\sigma_1),(2,\sigma_2)\) in order, the un-normalized quantum state
 \begin{equation}
 \Omega=p_1^2p_2\, \rho_1\otimes \rho_1\otimes \rho_2+ p_2^2p_1 \,\rho_2\otimes \rho_2\otimes\rho_1,
 \label{Omega_example}
 \end{equation}
 which sums over two compatible sequences of intrinsic state labels $\mathbf i_1=(1,1,2), \mathbf i_2=(2,2,1)$. Indeed, due to permutation-invariance, the sample labels are only informative in the equality pattern.

 Similar to the above example, we can develop the formal description for \(M\) samples from the ensemble, with a collection of label-data pairs \(\{(\lambda_k,\sigma_k)\}_{k=1}^{M}\). The equality pattern of the sample labels can be described by a partition of the \(M\) samples into \(s\le N\) blocks
 \(
 \Pi=\{B_1,\ldots,B_s\}
 \)
 of the sample positions \([M]\equiv \{1,\ldots,M\}\), where
 \(t,t'\in B_a\)
 indicates that the labels of the \(t\)-th sample and \(t'\)-th sample are the same,
 \(\lambda_t=\lambda_{t'}\).
 For the example in Eq.~\eqref{Omega_example}, we have \(\Pi=\{\{1,2\},\{3\}\}\), since the first and second samples have the same label while the third has a distinct label. The intrinsic label of the sampled states is a $M$-length vector,
 $\mathbf i\in[N]^M$.
 When a sequence of intrinsic labels is compatible with the partition $\Pi$, we denote it as \(\operatorname{eq}(\mathbf i)=\Pi\). As the ensemble is label-invariant, the un-normalized quantum state representing the partition \(\Pi\) needs to sum over all possible sequences compatible with the partition,
 \begin{align}
 \Omega_{\Pi}
 ={}&
 \sum_{\mathbf i\in[N]^M:\operatorname{eq}(\mathbf i)=\Pi}
 \left[
 \prod_{t=1}^M
 p_{i_t}
 \right]
 \bigotimes_{t=1}^M
 \rho_{i_t} .
 \label{eq:Omega-def}
 \end{align}
 The un-normalized state has trace
 \(
 \Tr\Omega_{\Pi}=P(\Pi)
 =
 \sum_{\mathbf i\in[N]^M:\operatorname{eq}(\mathbf i)=\Pi}
 \prod_{t=1}^M p_{i_t},
 \)
 equaling the probability of observing such a partition in random-access ensemble. 

 After summing over compatible intrinsic-label assignments, one can show that the un-normalized quantum state describing the sampling process can be expressed by the power-weighted moment operator (see Appendix A)
 \begin{equation}
 \mathcal{K}^{(r)}
 :=
 \sum_{i=1}^N p_i^r \rho_i^{\otimes r}.
 \label{eq:power-moment}
 \end{equation}
 For example, given the state in Eq.~\eqref{Omega_example} for the partition \(\Pi=\{\{1,2\},\{3\}\}\), we have
 \begin{align}
 \Omega_{\{\{1,2\},\{3\}\}}&
 =\mathcal{K}^{(2)}\otimes \mathcal{K}^{(1)}
 -
 \mathcal{K}^{(3)}.
 \end{align}
 The moment-based expression in the general case requires the evaluation of M\"obius coefficients and can be found in Appendix A.

{\em Hypothesis testing between ensembles.---}
With the quantum description of ensemble sampling in hand, we consider the binary hypotheses \(H_h:\cE_h=\{(p_i^{(h)},\rho_i^{(h)})\}_{i=1}^N\), \(h\in\{0,1\}\), with priors \(\eta_0,\eta_1>0\) satisfying \(\eta_0+\eta_1=1\). Given \(M\) samples drawn from one of the two ensembles, the tester uses the observed equality pattern and a joint measurement on the sampled quantum systems to infer the hypothesis \(h\). We append the subscript \(h\) to quantities associated with \(\cE_h\); in particular, \(\Omega_{h,\Pi}\) denotes the branch state in Eq.~\eqref{eq:Omega-def} under \(H_h\). We first derive the fundamental limits on the Bayes error (see Appendix~\ref{sec:proof_theorem1} for a proof).
\begin{theorem}[Minimum error probability for hypothesis testing between quantum ensembles]
\label{thm:opt}
The minimum Bayes error
\begin{equation}
P_{e,{\rm opt}}^{(M)}
=
\frac12\left[
1-
\sum_{\Pi\in\Part([M])}
\norm{\eta_0\Omega_{0,\Pi}-\eta_1\Omega_{1,\Pi}}_1
\right].
\label{eq:Bayes-exact}
\end{equation}
It is achieved by the positive spectral projector of $\eta_0\Omega_{0,\Pi}-\eta_1\Omega_{1,\Pi}$
\end{theorem}

In the classical limit, where two ensembles are composed of the same orthonormal basis states, $\rho_i=\ketbra{i}$, the testing reduces to testing between two classical probability distributions, and the exact Bayes error reduces to the well-known classical formula,
$
P_{e,{\rm opt},c}^{(M)} = \left[1-\|\eta_0[{\bm p}^{(0)}]^{\otimes M}-\eta_1[{\bm p}^{(1)}]^{\otimes M}\|_1\right]/2 $,
where ${\bm p}^{\otimes M}$ denotes the $M$-fold product distribution.



In the many-sample regime, we consider the error exponent \(E:=\lim_{M\to\infty}-\log\left( P_{e,\mathrm{opt}}^{(M)}\right)/M\). Writing \(Q_s(A,B):=\Tr(A^sB^{1-s})\), define the fixed-permutation Chernoff information~\cite{audenaert} and the corresponding across-hypothesis and in-hypothesis permutation bottlenecks by
\begin{align}
\xi_\pi^{h,h'}&:=-\log\min_{0\le s\le1}\sum_{i=1}^N Q_s\!\left(p_i^{(h)}\rho_i^{(h)},p_{\pi(i)}^{(h')}\rho_{\pi(i)}^{(h')}\right),\notag\\
\xi_{\mathrm{acr}}&:=\min_{\pi\in S_N}\xi_\pi^{0,1},\qquad \xi_{\mathrm{in}}:=\min_{h\in\{0,1\}}\min_{\pi\in S_N\setminus\{\mathrm{id}\}}\xi_\pi^{h,h}.
\label{eq:generalized-permutation-chernoff}
\end{align}
We have the following theorem (see Appendix~\ref{app:proof_theorem2} for a proof).
\begin{theorem}[Asymptotic Chernoff bounds for optimal Bayes error]
\label{thm:fixed-N-chernoff-bounds-main}
The error exponent from Eq.~\eqref{eq:Bayes-exact} satisfies
\begin{equation}
\min\{\xi_{\mathrm{in}},\xi_{\mathrm{acr}}\}\le E\le\xi_{\mathrm{acr}}.
\end{equation}
For finite uniform pure-state ensembles, the upper bound is tight and \(E=\xi_{\mathrm{acr}}\).
\end{theorem}

Below, we illustrate some examples of hypothesis testing between ensembles. We begin with the example of testing between state designs. Then, we provide additional examples in connection to constellation in optical communication.


\begin{figure}[t]
    \centering
    \includegraphics[width=\linewidth]{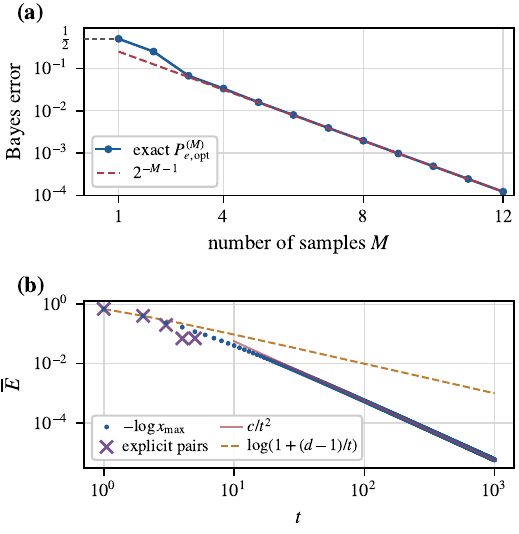}
    \caption{Results in testing between designs. (a) Exact Bayes error for testing between $\mathcal{E}Z$ and $\mathcal{E}X$ with equal priors, with $P_{e,\mathrm{opt}}^{(1)}=1/2$ indicated. The dashed line shows the asymptotic behavior $P_{e,\mathrm{opt}}^{(M)}\sim2^{-M-1}$. (b) Universal upper bound $\overline{E}$ for testing between qubit $t$-design ensembles. Crosses show the exponents attained by the explicit design pairs for $t=1,\ldots,5$, while the reference curve $5.678/t^2$ illustrates the asymptotic scaling $\overline{E}=O(t^{-2})$. The weaker estimate \(\log[1+(d-1)/t]\) is also shown.}
    
    \label{fig:results in testing between designs}
\end{figure}

\emph{Testing between state designs.---} An important family of quantum ensembles is the $t$-design, where the prior is uniform and the $r$-th moment operator $\calM^{(r)}$ matches that of Haar random ensemble for all $r\le t$. For simplicity in describing uniform ensembles, we introduce the moment operator
$
\calM^{(r)}
:=\expval{\rho^{\otimes r}}
$, which connects to the power-weighted moments directly as 
$
\mathcal{K}^{(r)}
=
N^{1-r}\calM^{(r)}.
$

In general, the construction of $t$-design ensemble is not unique. For example, $\calE_Z=\{\ket{0},\ket{1}\}$ and $\calE_X=\{\ket{+},\ket{-}\}$ are both $1$-designs. And they are different at the second moment $\calM^{(2)}$. A hypothesis testing task between different realizations of $t$-design naturally arises.
Such a hypothesis testing task motivates the following theorem for distinguishing two uniform ensembles with all $r<k$ moments equal, but with different $r=k$-th moments. In the equal prior case, we obtain the following (see Appendix~\ref{app:extra_proofs} for a proof).
\begin{corollary}[Testing when lower moments match]\label{cor:first-informative-moment}
Consider two uniform ensembles with equal priors and suppose \(M_0^{(r)}=M_1^{(r)}\) for \(1\le r<k\). Then \(P_{e,\mathrm{opt}}^{(M)}=1/2\) for \(M<k\), and
\[
P_{e,\mathrm{opt}}^{(k)}=\frac12\left[1-\frac12A_kN^{1-k}\|M_0^{(k)}-M_1^{(k)}\|_1\right]
\]
for $M=k$, where $A_k:=\sum_{\ell=1}^k\left\{{k\atop\ell}\right\}(\ell-1)!$.
\end{corollary}

For \(\mathcal E_Z=\{\ket0,\ket1\}\) and \(\mathcal E_X=\{\ket+,\ket-\}\), the first informative moment is \(k=2\). Since \(A_2=2\) and \(\|M_Z^{(2)}-M_X^{(2)}\|_1=1\), Corollary~\ref{cor:first-informative-moment} gives \(P_{e,\mathrm{opt}}^{(2)}=1/4\). The exact $M>k$ expression is shown as Eq.~\eqref{eq:ZX-error} in Appendix~C and satisfies \(P_{e,\mathrm{opt}}^{(M)}\sim\eta_Z\eta_X 2^{1-M}\). We show the exact Bayes error when $M$ increases in Fig.~\ref{fig:results in testing between designs}(a).



In the general case, however, the expressions for $M>k$ under Corollary~\ref{cor:first-informative-moment} become increasingly complicated. Nevertheless, we are able to obtain
an upper bound on the error exponent (see Appendix~\ref{app:extra_proofs} for a proof).
\begin{theorem}[Universal bounds for testing between two $t$-designs]\label{thm: unversal chernoff t-design upper bound}
Suppose in a binary hypothesis testing task, the two ensembles are both $t$-design pure uniform ensembles. Let $x_{\rm max}$ be the largest zero in $(0,1)$ of the Jacobi polynomial $P_{\lceil t/2\rceil}^{(d-2,\beta_t)}(2x-1)$, where $\beta_t=0$ for odd $t$ and $\beta_t=1$ for even $t$. Then the error exponent
\begin{equation}
E\le-\log x_{\max}=O_d(t^{-2})
\end{equation}
for fixed \(d\) as \(t\to\infty\).
\end{theorem}

For \(d=2\), the numerical values of \(-\log x_{\max}\) are shown in Fig.~\ref{fig:results in testing between designs}(b). The large-\(t\) data exhibit the \(t^{-2}\) scaling, as illustrated by the reference curve \(c/t^2\). The crosses show the exponents attained by the explicit \(t\)-design pairs for \(t=1,\ldots,5\). The weaker elementary bound \(\log[1+(d-1)/t]\) is also shown; it is tight for \(t=1,2\), but becomes asymptotically looser because it scales as \(O(t^{-1})\).

A corresponding finite-sample converse follows (see Appendix~\ref{app:extra_proofs} for a proof).
\begin{corollary}[Finite-sample converse for testing state designs]
Let \(\mathcal E_h=\{(1/N,\rho_i^{(h)})\}_{i=1}^N\), \(h=0,1\), be two finite uniform pure-state \(t\)-designs on \(\mathbb C^d\), and let \(x_{\max}\) be defined as in Theorem~\ref{thm: unversal chernoff t-design upper bound}. Then, for every \(M\ge 1\),
\begin{equation}
P_{e,\mathrm{opt}}^{(M)}\ge\frac12\left[1-\sqrt{1-4\eta_0\eta_1x_{\max}^{M}}\right].
\end{equation}
Let $M_\epsilon:=\inf\{M\in\mathbb N:P_{e,\mathrm{opt}}^{(M)}\le\epsilon\}$.  Consequently, $M_\epsilon=\Omega_{d,\epsilon}(t^2)$ as $t\to\infty$ for fixed $d$ and $\epsilon$.
\end{corollary}

For every fixed $d\ge2$, there exists a constant $c_d>0$, independent of $t$, such that for each integer $t\ge1$ there exist two finite uniform pure-state $t$-designs of equal cardinality on $\mathbb C^d$ whose Chernoff exponent obeys \(E\ge-\log(1-c_dt^{-2})\ge c_dt^{-2}\). Thus, for every fixed $d$, the supremum of the Chernoff exponent over all such pairs scales as $\Theta_d(t^{-2})$ as $t\to\infty$.

\emph{Symmetric two-component ensembles and optical communication.---}
\label{subsec:symmetric-two-component}Consider two uniform pure-state ensembles
\(
\mathcal E_h=\left\{\left(\frac12,\ketbra*{\psi^{(h)}_1}\right),\left(\frac12,\ketbra*{\psi^{(h)}_2}\right)\right\},
\)
with $h=0,1$ and hypothesis priors \(\eta_0,\eta_1>0\) and \(\eta_0+\eta_1=1\). Suppose that
$\braket*{\psi^{(0)}_1}{\psi^{(0)}_2}=\braket*{\psi^{(1)}_1}{\psi^{(1)}_2}=x\in\mathbb R,~|x|<1$, and $
\braket*{\psi^{(0)}_1}{\psi^{(1)}_1}=\braket*{\psi^{(0)}_2}{\psi^{(1)}_2}=y,~ \braket*{\psi^{(0)}_1}{\psi^{(1)}_2}=\braket*{\psi^{(0)}_2}{\psi^{(1)}_1}=y^\star$. For this symmetry class, Eq.~\eqref{eq:Bayes-exact} admits a closed form Eq.~\eqref{PE_symmetric_two} in Appendix~C, with \(P_{e,\mathrm{opt}}^{(M)}=2\eta_0\eta_1|y|^{2M}[1+o(1)]\) and \(E=-2\log|y|\) in the $M\gg1$ limit.


In a communication-system design, the bit values \(0\) and \(1\) serve only as labels for the two signal states. In binary phase-shift keying (BPSK), one may assign the coherent states \(\ket{\alpha}\) and \(\ket{-\alpha}\) to bits \(0\) and \(1\), respectively, or reverse this assignment. Since exchanging the two bit labels does not change the physical signal constellation, both assignments are represented by the same label-invariant ensemble
\(
\mathcal E_R=\left\{\left(\frac12,\ketbra{\alpha}\right),\left(\frac12,\ketbra{-\alpha}\right)\right\}.
\)
Applying an optical phase shift of \(\pi/2\) to the signal states gives the rotated constellation
\(
\mathcal E_I=\left\{\left(\frac12,\ketbra{-i\alpha}\right),\left(\frac12,\ketbra{i\alpha}\right)\right\}.
\)
Thus, testing between \(\mathcal E_R\) and \(\mathcal E_I\) determines which constellation orientation is implemented while remaining insensitive to an exchange of the bit labels.

Let \(\bar n:=|\alpha|^2\) be the average photon number of each coherent signal. The overlap between the two states within either constellation is
\(
\left|\braket{\alpha}{-\alpha}\right|^2=\left|\braket{-i\alpha}{i\alpha}\right|^2=e^{-4\bar n},
\)
whereas every overlap between a state in \(\mathcal E_R\) and a state in \(\mathcal E_I\) has squared modulus
\(
\left|\braket{\pm\alpha}{\pm i\alpha}\right|^2
=\left|\braket{\pm\alpha}{\mp i\alpha}\right|^2
=e^{-2\bar n}.
\)
With the displayed ordering of the states in \(\mathcal E_I\), the symmetry parameters are
\(
x=e^{-2\bar n},\qquad y=e^{-(1+i)\bar n}
\)
and the corresponding minimum Bayes error is given by Eq.~\eqref{eq:coherent-BPSK-error} shown in Appendix~C with the above $x,y$ values.
In the $M\gg1$ limit, the Chernoff exponent
\(
E=-2\log|y|=2|\alpha|^2.
\)

\emph{Two common pure component states.---}
Denote an ensemble of pure states as
$\mathcal E_p=\{(p,\rho_1),(1-p,\rho_2)\}$, where
$\rho_1=|\psi_1\rangle\langle\psi_1|$, $\rho_2=|\psi_2\rangle\langle\psi_2|$ are pure states with overlap $c=|\langle\psi_1|\psi_2\rangle|^2$.
We consider two hypotheses with the same pure component states but different probabilities:
\begin{align}
H_0: \mathcal E=\mathcal E_p; H_1:  \mathcal E =\mathcal E_q.
\end{align}
The exact equal-prior Bayes error obtained from Eq.~\eqref{eq:Bayes-exact} can be solved in a closed form that sums over $M/2+1$ terms (see Eq.~\eqref{eq:two-state-common-pure-bayesm} of Appendix~C).

\begin{figure}
    \centering
    \includegraphics[width=\linewidth]{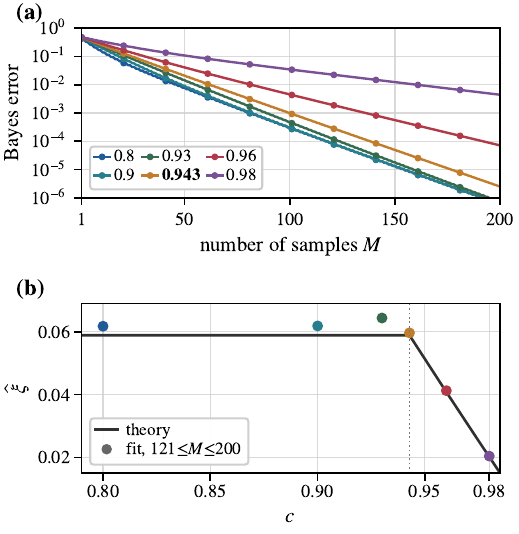}
    \caption{Results for the swapped-probability ensembles with $p=1/3$. (a) Exact Bayes error for several state overlaps $c=\lvert\langle\psi_1|\psi_2\rangle\rvert^2$ as the number of samples $M$ increases. (b) The theoretical exponent $\xi(c)=\min{\log[3/(2\sqrt{2})],-\log c}$ and finite-$M$ estimates $\widehat{\xi}$ fitted over $121\leq M\leq200$. The dotted line marks the transition at $c_*=2\sqrt{2}/3=0.943$.}
    \label{fig: results in testing swap two common}
\end{figure}


For this family, the Chernoff upper bound in Theorem~\ref{thm:fixed-N-chernoff-bounds-main} is tight. Define the Bernoulli Chernoff information as $\xi_{\mathrm B}(u,v)=-\log\min_{0\le s\le1}[u^sv^{1-s}+(1-u)^s(1-v)^{1-s}]$. We have
\begin{equation}
E=\xi_{\mathrm{acr}} = \min\{\xi_{\mathrm B}(p,q),-\log c+\xi_{\mathrm B}(p,1-q)\}.
\label{eq:two-common-pure-exponent}
\end{equation}

A simple subcase is the swapped-probability ensemble, with $q=1-p$.
For this case,
\(
E=\min\{\xi_{\mathrm B}(p,1-p),-\log c\}.
\)
The two terms correspond to the classical distinguishability of the probability distributions and the quantum distinguishability of the component states, respectively. As shown in Fig.~\ref{fig: results in testing swap two common}, taking $p=1/3$ the error
exponent undergoes a transition as the overlap parameter $c$ crosses
the critical value, revealing that the testing is controlled by both classical probabilities and quantum-state overlaps. 

\emph{Discussions.---} In this work, we focus primarily on the random-access model, in which each query returns $\rho_i$ with probability $p_i$. One may also consider a controlled-access model, in which the tester can choose the labels to query while the corresponding states remain unknown. This models ensembles from the embeddings of classical data into quantum states for quantum machine learning~\cite{schuld2019quantum,havlicek2019supervised,lloyd2020quantum,huang2021power}.
Our analysis extends directly to this setting by optimizing over $\Pi\in\mathrm{Part}([M])$ so as to minimize the Bayes error in Eq.~\eqref{eq:Bayes-exact}.

Beyond the asymptotic regime $M\to\infty$, one may also study the sample complexity, namely, the number of samples required to distinguish two ensembles with a fixed constant error probability. For uniform pure-state ensembles, techniques similar to those in Ref.~\cite{yao2026hierarchydiscriminativepowercomplexity} yield a sample complexity that scales as $N^{1-1/k}$ as $N$ increases, where $k$ is the lowest moment order at which the two ensembles differ.

\begin{acknowledgements}
J.Y. and Q.Z. acknowledge support from Office of Naval Research (N00014-23-1-2296, MURI N000142612102), DARPA (HR0011-24-9-0362,HR00112490453,D24AC00153-02), NSF (OMA-2326746, 2350153, CCF-2240641), AFOSR MURI FA9550-24-1-0349, Halliburton Company and an unrestricted gift from Google. OpenAI ChatGPT (GPT-5.6) is used to assist the analyzes and authors have verified the results. The authors take full responsibility for the contents of this work.
\end{acknowledgements}



%

\

{\em Appendix A: Moment representation of the quantum ensemble state.---}\label{app:A moment and mobuis}
Given an ensemble of states, one can define the power-weighted moment operator
\begin{equation}
\mathcal{K}^{(r)}
:=
\sum_{i=1}^N p_i^r \rho_i^{\otimes r}.
\label{eq:power-moment}
\end{equation}
For uniform ensembles, the power-weighted moments go back to those adopted in defining $t$-design
$
\mathcal{K}^{(r)}
=
N^{1-r}\calM^{(r)},
$
where the moment operator
$\calM^{(r)}:=\expval{\rho^{\otimes r}}$.
 
We define \(R\succeq\Pi\) if \(R\) is a coarsening of \(\Pi\), namely if every block \(D\in R\) is a union of blocks of \(\Pi\). For example, $\Pi^\prime=\{B_1,B_2B_3\}$ is the coarsening of the partition $\Pi=\{B_1,B_2,B_3\}$, $\Pi^\prime\succeq\i $. Then, we present the following theorem (see Appendix~\ref{app:mobius-proof} for a proof).
\begin{theorem}[M\"obius moment expansion]
\label{thm:mobius}
For every observed branch \(\Pi\in\mathrm{Part}([M])\),
\begin{equation}
\Omega_{\Pi}
=
\sum_{R\succeq\Pi}
\mu(\Pi,R) \bigotimes_{D\in R} \mathcal{K}^{(|D|)},
\label{eq:Mobius-general}
\end{equation}
where the M\"obius coefficient \cite{Rota1964OnTFmobius} between \(\Pi\) and \(R\) is
\(
\mu(\Pi,R)
=
\prod_{D\in R}
(-1)^{m_D-1}(m_D-1)! .
\label{eq:mobius-coeff}
\)
Here \(m_D\) is the number of blocks of \(\Pi\) contained in the block \(D\in R\). 
For uniform ensembles,
\begin{equation}
\Omega_{\Pi}
=
\sum_{R\succeq\Pi}
\mu(\Pi,R)N^{|R|-M} \bigotimes_{D\in R} \calM^{(|D|)} .
\label{eq:Mobius-uniform}
\end{equation}
\end{theorem}

In the above theorem, each factor \(\mathcal{K}^{(|D|)}\) or \(M^{(|D|)}\) is understood to act on the sample registers indexed by the block \(D\). 

Applying Theorem~\ref{thm:mobius} yields explicit moment expansions for larger \(M\). For \(M=3\), first consider the branch \(\Pi=\{\{1,2\},\{3\}\}\), in which the first two samples have the same label and the third has a different label. Its only coarsenings are \(R_0=\{\{1,2\},\{3\}\}\) and \(R_1=\{\{1,2,3\}\}\), with M\"obius coefficients \(1\) and \(-1\), respectively. Hence
\begin{align}
\Omega_{\{\{1,2\},\{3\}\}}&=\sum_{i\ne j}p_i^2p_j\,\rho_i^{(1)}\otimes\rho_i^{(2)}\otimes\rho_j^{(3)}\notag\\
&=\mathcal{K}^{(2)}_{\{1,2\}}\otimes\mathcal{K}^{(1)}_{\{3\}}-\mathcal{K}^{(3)}_{\{1,2,3\}},
\end{align}
where the subscript \(B\) in \(\mathcal{K}^{(r)}_B\) specifies the sample registers on which the moment operator acts.

For the branch \(\Pi=\{\{1\},\{2\},\{3\}\}\), in which all three labels are distinct, the coarsening formula gives
\begin{align}
&\Omega_{\{1|2|3\}}
\nonumber
\\
=&\mathcal{K}^{(1)}_{\{1\}}\otimes\mathcal{K}^{(1)}_{\{2\}}\otimes\mathcal{K}^{(1)}_{\{3\}}-\mathcal{K}^{(2)}_{\{1,2\}}\otimes\mathcal{K}^{(1)}_{\{3\}}-\mathcal{K}^{(2)}_{\{1,3\}}\otimes\mathcal{K}^{(1)}_{\{2\}}\notag\\
&-\mathcal{K}^{(2)}_{\{2,3\}}\otimes\mathcal{K}^{(1)}_{\{1\}}+2\mathcal{K}^{(3)}_{\{1,2,3\}}.
\end{align}
Equivalently,
\[
\Omega_{\{1|2|3\}}=\sum_{\substack{i,j,k\\\mathrm{all\ distinct}}}p_ip_jp_k\,\rho_i^{(1)}\otimes\rho_j^{(2)}\otimes\rho_k^{(3)}.
\]
The coefficient \(2\) in the last term is the M\"obius coefficient \(\mu(\{\{1\},\{2\},\{3\}\},\{\{1,2,3\}\})=2\).

Theorem~\ref{thm:mobius} also allows the optimal error probability to be evaluated directly from the ensemble moments. For uniform ensembles with equal priors, substituting Eq.~\eqref{eq:Mobius-uniform} into Eq.~\eqref{eq:Bayes-exact} gives the following exact expression. Define
\begin{equation}
\Xi_\Pi^{(M)}:=\sum_{R\succeq\Pi}\mu(\Pi,R)N^{|R|-M}\left(M_0^R-M_1^R\right).
\label{eq:Xi-def}
\end{equation}
Then
\begin{equation}
P_{e,\mathrm{opt}}^{(M)}=\frac12\left[1-\frac12\sum_{\Pi\in\Part([M])}\norm{\Xi_\Pi^{(M)}}_1\right].
\label{eq:Bayes-moment-lattice}
\end{equation}
Thus, the optimal Bayes error is exactly a trace-norm functional of the full moment hierarchy up to order \(M\), rather than of any single moment.

{\em Appendix B: Connection to Wasserstein distance.---}\label{app:B Wasserstein an exponent} To relate the Chernoff exponent to the geometry of quantum ensembles, recall that the Wasserstein distance used in this work is the optimal-transport cost between \(\calE_0=\{(p_i^{(0)},\rho_i^{(0)})\}_{i=1}^N\) and \(\calE_1=\{(p_j^{(1)},\rho_j^{(1)})\}_{j=1}^N\), computed by the linear program \cite{Wasserstein_10.1561/2200000073}
\[
W(\calE_0,\calE_1):=\min_{\gamma\in\Gamma(p^{(0)},p^{(1)})}\sum_{i,j=1}^N\gamma_{ij}C_{ij},
\]
where
\[
\Gamma(p^{(0)},p^{(1)}):=\left\{\gamma\in\mathbb R_+^{N\times N}:\sum_{j=1}^N\gamma_{ij}=p_i^{(0)},\ \sum_{i=1}^N\gamma_{ij}=p_j^{(1)}\right\}
\]
is the set of transport plans with marginals \(p^{(0)}\) and \(p^{(1)}\). For pure-state ensembles, we use the fidelity cost \(C_{ij}:=1-\Tr(\rho_i^{(0)}\rho_j^{(1)})\). Under the natural identification of ensembles as weighted measures on quantum states, this Wasserstein distance has full discriminative power, namely \(W(\calE_0,\calE_1)=0\) if and only if \(\calE_0=\calE_1\) \cite{yao2026hierarchydiscriminativepowercomplexity}.

For finite uniform pure-state ensembles of equal cardinality, the Birkhoff--von Neumann theorem gives
\[
W(\calE_0,\calE_1)=\min_{\pi\in S_N}\frac1N\sum_{i=1}^NC_{i,\pi(i)}=1-F_\star,
\]
where \(F_\star:=\max_{\pi\in S_N}N^{-1}\sum_{i=1}^N\Tr(\rho_i^{(0)}\rho_{\pi(i)}^{(1)})\). Since \(E=\xi_{\mathrm{acr}}=-\log F_\star\) for this class of ensembles, we obtain
\begin{equation}
    E=\xi_{\mathrm{acr}}=-\log\!\left[1-W(\calE_0,\calE_1)\right],
\end{equation}
or equivalently \(W(\calE_0,\calE_1)=1-e^{-E}=1-e^{-\xi_{\mathrm{acr}}}\). This relation is physically natural: a smaller optimal transport cost corresponds to a closer componentwise matching and hence a smaller discrimination exponent.


{\em Appendix C: Exact formula for the examples---} For the testing between $\calE_Z$ and $\calE_X$, the exact Bayes error is
\begin{equation}
P_{e,{\rm opt}}^{(M)}
=
\begin{cases}
\dfrac14\left[1-\sqrt{1-4\eta_Z\eta_X\,2^{2-M}}\right],
& M\ {\rm even}, \\[1.2ex]
\dfrac12\left[1-\sqrt{1-4\eta_Z\eta_X\,2^{1-M}}\right],
& M\ {\rm odd}.
\end{cases}
\label{eq:ZX-error}
\end{equation}

For testing between symmetric two-component ensembles, the exact Bayes err or is
\begin{equation}
P_{e,\mathrm{opt}}^{(M)}
=\frac12\Bigg[1 -\sum_{s=\pm1}\frac{1+s x^M}{2}\sqrt{1-4\eta_0\eta_1\left|\frac{y^M+s(y^\star)^M}{1+sx^M}\right|^2}\Bigg].
\label{PE_symmetric_two}
\end{equation}
In the example of optical communication system. The exact Bayes error is
\begin{equation}
\begin{aligned}
P_{e,\mathrm{opt}}^{(M)}
=\frac12\Bigg[1
&-\frac{1+e^{-2z}}{2}
\sqrt{1-4\eta_0\eta_1\frac{\cos^2z}{\cosh^2z}}\\
&-\frac{1-e^{-2z}}{2}
\sqrt{1-4\eta_0\eta_1\frac{\sin^2z}{\sinh^2z}}
\Bigg].
\end{aligned}
\label{eq:coherent-BPSK-error}
\end{equation}
For testing between ensembles with two common pure component states, the exact Bayes error is
\begin{equation}
P_{e,\mathrm{opt}}^{(M)}
=
\frac12\left[
1-\frac12\sum_{r=0}^{\lfloor M/2\rfloor}
\nu_{M,r}\Phi_{c^M}(A_{r,M-r},B_{r,M-r})
\right],
\label{eq:two-state-common-pure-bayesm}
\end{equation}
where $A_{r,M-r}=p^r(1-p)^{M-r}-q^r(1-q)^{M-r}$, $B_{r,M-r}=(1-p)^rp^{M-r}-(1-q)^rq^{M-r}$, and
\[
\Phi_z(A,B)=
\begin{cases}
|A+B|,&AB\ge0,\\
\sqrt{(A-B)^2+4ABz},&AB<0.
\end{cases}
\]
Here $\nu_{M,r}$ is defined by
\[
\nu_{M,r}:=\begin{cases}1,&r=0,\\[1mm]\binom Mr,&1\le r<M/2,\\[1mm]\frac12\binom M{M/2},&r=M/2\text{ and }M\text{ is even}.\end{cases}
\]

\clearpage
\onecolumngrid
\appendix

\tableofcontents

\section{M\"obius expansion and proof of Theorem~\ref{thm:mobius}}
\label{app:mobius-proof}

For \(\mathbf i=(i_1,\ldots,i_M)\in[N]^M\), define
\begin{equation}
X(\mathbf i):=\left(\prod_{t=1}^Mp_{i_t}\right)\bigotimes_{t=1}^M\rho_{i_t}^{(t)}.
\label{eq:sequence-state}
\end{equation}
For \(S,R\in\Part([M])\), define the exact equality-pattern sum and the diagonal sum, respectively, by
\begin{equation}
G_S:=\sum_{\mathbf i:\eqp(\mathbf i)=S}X(\mathbf i),\qquad H_R:=\sum_{\mathbf i:\eqp(\mathbf i)\succeq R}X(\mathbf i).
\label{eq:exact-and-diagonal-sums}
\end{equation}
In particular, \(\Omega_\Pi=G_\Pi\). We use \(S\succeq R\) to mean that \(S\) is coarser than or equal to \(R\), and \(S\succ R\) for strict coarsening. For example, if \(M=3\), \(R=\{\{1\},\{2\},\{3\}\}=\hat 0\), and \(S=\{\{1,2\},\{3\}\}\), then \(S\succ R\) because every block of \(R\) is contained in a block of \(S\). Since every sequence \(\mathbf i\) has a unique equality partition,
\begin{equation}
H_R=\sum_{S\succeq R}G_S.
\label{eq:zeta-transform}
\end{equation}
Equation~\eqref{eq:zeta-transform} is the zeta transform on the finite poset \(\Part([M])\). M\"obius inversion \cite{Rota1964OnTFmobius} therefore gives
\begin{equation}
G_\Pi=\sum_{R\succeq\Pi}\mu(\Pi,R)H_R.
\label{eq:mobius-inversion-G}
\end{equation}

We next identify \(H_R\). Let \(R=\{D_1,\ldots,D_\ell\}\). The condition \(\eqp(\mathbf i)\succeq R\) requires all positions within each block \(D\in R\) to carry the same intrinsic component label, while labels assigned to different blocks are summed independently and need not be distinct. By contrast, if \(\eqp(\mathbf i)=R\), distinct blocks must carry distinct labels. Hence
\begin{align}
H_R&=\sum_{(j_D)_{D\in R}\in[N]^{|R|}}\left(\prod_{D\in R}p_{j_D}^{|D|}\right)\bigotimes_{D\in R}\bigotimes_{t\in D}\rho_{j_D}^{(t)}\notag\\
&=\bigotimes_{D\in R}\left(\sum_{j=1}^Np_j^{|D|}\bigotimes_{t\in D}\rho_j^{(t)}\right)=\bigotimes_{D\in R}\mathcal K^{(|D|)},
\label{eq:diagonal-sum-factorization}
\end{align}
where each moment operator acts on the registers indexed by \(D\). Substituting Eq.~\eqref{eq:diagonal-sum-factorization} into Eq.~\eqref{eq:mobius-inversion-G} and using \(\Omega_\Pi=G_\Pi\) yields
\[
\Omega_\Pi=\sum_{R\succeq\Pi}\mu(\Pi,R)\bigotimes_{D\in R}\mathcal K^{(|D|)},
\]
which is Eq.~\eqref{eq:Mobius-general}.

It remains to determine the M\"obius coefficient explicitly. For \(R\succeq\Pi\) and \(D\in R\), define
\begin{equation}
m_D:=\#\{B\in\Pi:B\subseteq D\},
\label{eq:number-merged-blocks}
\end{equation}
the number of blocks of \(\Pi\) merged into \(D\). The interval \([\Pi,R]\) decomposes into independent partition lattices associated with the blocks of \(R\): for each \(D\in R\), the corresponding factor is the partition lattice of the \(m_D\) blocks of \(\Pi\) contained in \(D\). Thus
\begin{equation}
[\Pi,R]\cong\prod_{D\in R}\Pi_{m_D},
\label{eq:partition-interval-product}
\end{equation}
where \(\Pi_{m_D}\) denotes the partition lattice of a set with \(m_D\) elements. Using multiplicativity of the M\"obius function under direct products and the standard identity \(\mu_{\Pi_m}(\hat 0,\hat 1)=(-1)^{m-1}(m-1)!\) \cite{Rota1964OnTFmobius}, we obtain
\begin{equation}
\mu(\Pi,R)=\prod_{D\in R}(-1)^{m_D-1}(m_D-1)!.
\label{eq:explicit-mobius-coefficient}
\end{equation}
This proves Eq.~\eqref{eq:Mobius-general}.

For a uniform ensemble, \(p_i=1/N\), the power-weighted and ordinary moment operators satisfy
\begin{equation}
\mathcal K^{(|D|)}=N^{1-|D|}\calM^{(|D|)},\qquad \calM^{(|D|)}:=\frac1N\sum_{i=1}^N\bigotimes_{t\in D}\rho_i^{(t)}.
\label{eq:uniform-power-ordinary-relation}
\end{equation}
Therefore,
\begin{equation}
\bigotimes_{D\in R}\mathcal K^{(|D|)}=N^{\sum_{D\in R}(1-|D|)}\bigotimes_{D\in R}\calM^{(|D|)}=N^{|R|-M}\bigotimes_{D\in R}\calM^{(|D|)},
\label{eq:uniform-moment-product}
\end{equation}
where \(\sum_{D\in R}|D|=M\). Substituting Eq.~\eqref{eq:uniform-moment-product} into Eq.~\eqref{eq:Mobius-general} gives
\[
\Omega_\Pi=\sum_{R\succeq\Pi}\mu(\Pi,R)N^{|R|-M}\bigotimes_{D\in R}\calM^{(|D|)},
\]
which is Eq.~\eqref{eq:Mobius-uniform} and completes the proof of Theorem~\ref{thm:mobius}.

\section{Proof of Theorem~\ref{thm:opt} and the classical limit}
\label{sec:proof_theorem1}

\begin{proof}[Proof of Theorem~\ref{thm:opt}]
For each hypothesis \(h\in\{0,1\}\) and equality pattern \(\Pi\in\Part([M])\), recall the unnormalized branch state
\[
\Omega_{h,\Pi}:=\sum_{\substack{\mathbf i\in[N]^M:\\\operatorname{eq}(\mathbf i)=\Pi}}\left(\prod_{t=1}^Mp_{i_t}^{(h)}\right)\bigotimes_{t=1}^M\rho_{i_t}^{(h)},
\]
whose trace \(\Tr\Omega_{h,\Pi}\) is the probability of observing \(\Pi\) under \(H_h\). For a branch-dependent binary POVM \(\{E_{0|\Pi},E_{1|\Pi}\}\), define the Helstrom operator
\begin{equation}
D_\Pi:=\eta_0\Omega_{0,\Pi}-\eta_1\Omega_{1,\Pi}.
\label{eq:branch-Helstrom-operator}
\end{equation}
Using \(E_{1|\Pi}=I-E_{0|\Pi}\), the contribution of branch \(\Pi\) to the unconditional Bayes error is
\begin{equation}
R_\Pi(E_{0|\Pi}):=\eta_0\Tr(E_{1|\Pi}\Omega_{0,\Pi})+\eta_1\Tr(E_{0|\Pi}\Omega_{1,\Pi})=\eta_0\Tr\Omega_{0,\Pi}-\Tr(E_{0|\Pi}D_\Pi).
\label{eq:branch-Bayes-risk}
\end{equation}
Thus, minimizing the branch error is equivalent to maximizing \(\Tr(E_{0|\Pi}D_\Pi)\) over all effects \(0\le E_{0|\Pi}\le I\).

Let \(D_\Pi=\sum_j\lambda_j\ket{j}\!\bra{j}\) be a spectral decomposition. For any effect \(0\le E\le I\),
\[
\Tr(ED_\Pi)=\sum_j\lambda_j\langle j|E|j\rangle\le\sum_{\lambda_j>0}\lambda_j=\Tr D_{\Pi,+},
\]
where \(D_{\Pi,+}\) is the positive part of \(D_\Pi\). Equality is attained by \(E=\Pi_+(D_\Pi)\), the projector onto the positive eigenspace of \(D_\Pi\). Hence
\[
E_{0|\Pi}^{\mathrm{opt}}=\Pi_+(D_\Pi),\qquad E_{1|\Pi}^{\mathrm{opt}}=I-\Pi_+(D_\Pi).
\]
Using \(\Tr D_{\Pi,+}=\frac12(\Tr D_\Pi+\|D_\Pi\|_1)\), the optimal branch error is
\begin{equation}
R_\Pi^{\mathrm{opt}}=\frac12\left[\eta_0\Tr\Omega_{0,\Pi}+\eta_1\Tr\Omega_{1,\Pi}-\|D_\Pi\|_1\right].
\label{eq:optimal-branch-Bayes-risk}
\end{equation}
Finally, since \(\sum_{\Pi\in\Part([M])}\Tr\Omega_{h,\Pi}=1\) for \(h=0,1\), summing Eq.~\eqref{eq:optimal-branch-Bayes-risk} over all equality patterns gives
\[
P_{e,\mathrm{opt}}^{(M)}=\sum_{\Pi\in\Part([M])}R_\Pi^{\mathrm{opt}}=\frac12\left[1-\sum_{\Pi\in\Part([M])}\|\eta_0\Omega_{0,\Pi}-\eta_1\Omega_{1,\Pi}\|_1\right],
\]
which proves Eq.~\eqref{eq:Bayes-exact}.
\end{proof}

We formally write down the classical limit as a corollary, then provide the prove.
\begin{corollary}[Classical limit]
    In classical limit, the exact Bayes error is
    \begin{equation}
    P_{e,{\rm opt},c}^{(M)} = \frac12\left[1-\|\eta_0[{\bm p}^{(0)}]^{\otimes M}-\eta_1[{\bm p}^{(1)}]^{\otimes M}\|_1\right] ,
    \label{eq:Bayes-exact-classical-supp}
    \end{equation}
    which recovers the well-knwon classical formula, where ${\bm p}^{\otimes M}$ is the tensor product of probability.\label{cor:classical limit Bayes error}
\end{corollary}
\begin{proof}[Proof of Corollary~\ref{cor:classical limit Bayes error}]
In the classical limit, the component states under both hypotheses are the same mutually orthogonal states, \(\rho_i^{(0)}=\rho_i^{(1)}=\ket{i}\!\bra{i}\), where \(\{\ket{i}\}_{i=1}^N\) is an orthonormal basis. Define \(\mathbf p^{(h)}:=(p_1^{(h)},\ldots,p_N^{(h)})\), \(\ket{\mathbf i}:=\bigotimes_{t=1}^M\ket{i_t}\), and \([\mathbf p^{(h)}]_{\mathbf i}^{\otimes M}:=\prod_{t=1}^Mp_{i_t}^{(h)}\). The branch states are diagonal in the common product basis:
\begin{equation}
\Omega_{h,\Pi}=\sum_{\substack{\mathbf i\in[N]^M:\\\operatorname{eq}(\mathbf i)=\Pi}}[\mathbf p^{(h)}]_{\mathbf i}^{\otimes M}\ket{\mathbf i}\!\bra{\mathbf i}.
\label{eq:classical-branch-state}
\end{equation}
Consequently,
\begin{equation}
\|\eta_0\Omega_{0,\Pi}-\eta_1\Omega_{1,\Pi}\|_1=\sum_{\substack{\mathbf i\in[N]^M:\\\operatorname{eq}(\mathbf i)=\Pi}}\left|\eta_0[\mathbf p^{(0)}]_{\mathbf i}^{\otimes M}-\eta_1[\mathbf p^{(1)}]_{\mathbf i}^{\otimes M}\right|.
\label{eq:classical-branch-trace-norm}
\end{equation}
Since the sets \(\{\mathbf i\in[N]^M:\operatorname{eq}(\mathbf i)=\Pi\}\), \(\Pi\in\Part([M])\), form a partition of \([N]^M\), summing Eq.~\eqref{eq:classical-branch-trace-norm} over \(\Pi\) yields
\begin{equation}
\sum_{\Pi\in\Part([M])}\|\eta_0\Omega_{0,\Pi}-\eta_1\Omega_{1,\Pi}\|_1=\left\|\eta_0[\mathbf p^{(0)}]^{\otimes M}-\eta_1[\mathbf p^{(1)}]^{\otimes M}\right\|_1.
\label{eq:classical-total-variation-reduction}
\end{equation}
Substituting Eq.~\eqref{eq:classical-total-variation-reduction} into Eq.~\eqref{eq:Bayes-exact} gives
\[
P_{e,\mathrm{opt},c}^{(M)}=\frac12\left[1-\left\|\eta_0[\mathbf p^{(0)}]^{\otimes M}-\eta_1[\mathbf p^{(1)}]^{\otimes M}\right\|_1\right],
\]
which is the standard classical Bayes error.
\end{proof}

\section{Proof of Theorem~\ref{thm:fixed-N-chernoff-bounds-main}}
\label{app:proof_theorem2}
We recall the definitions needed and results here. In the many-sample regime, we consider the error exponent \(E:=\lim_{M\to\infty}-\log\left( P_{e,\mathrm{opt}}^{(M)}\right)/M\). Writing \(Q_s(A,B):=\Tr(A^sB^{1-s})\), define the fixed-permutation Chernoff information and the corresponding across-hypothesis and in-hypothesis permutation bottlenecks by
\begin{align}
\xi_\pi^{h,h'}&:=-\log\min_{0\le s\le1}\sum_{i=1}^N Q_s\!\left(p_i^{(h)}\rho_i^{(h)},p_{\pi(i)}^{(h')}\rho_{\pi(i)}^{(h')}\right),\notag\\
\xi_{\mathrm{acr}}&:=\min_{\pi\in S_N}\xi_\pi^{0,1},\qquad \xi_{\mathrm{in}}:=\min_{h\in\{0,1\}}\min_{\pi\in S_N\setminus\{\mathrm{id}\}}\xi_\pi^{h,h}.
\label{eq:generalized-permutation-chernoff}
\end{align}

\begin{theorem}[Asymptotic Chernoff bounds for optimal Bayes error]
\label{thm:fixed-N-chernoff-bounds}
The error exponent $E$ has the following lower and upper bounds
\begin{equation}
    \min\{\xi_{\mathrm{in}},\xi_{\mathrm{acr}}\}\le E\le \xi_{\mathrm{acr}}.
\end{equation}

For testing between finite uniform pure state ensembles, we have
$
        E = \xi_{\text{acr}}.
$

Moreover, for finite uniform pure-state ensembles of equal cardinality, with fidelity cost
$C_{ij}=1-\Tr(\rho_i^{(0)}\rho_j^{(1)})$, the across-hypothesis Chernoff exponent and the Wasserstein distance satisfy
\[
    \xi_{\mathrm{acr}}
    =-\log\!\left[1-W(\mathcal E_0,\mathcal E_1)\right],
\]
or equivalently,
$W(\mathcal E_0,\mathcal E_1)=1-e^{-\xi_{\mathrm{acr}}}$.
\end{theorem}
\begin{proof}
    The proof contains four parts, first, we prove the lower bound, second, we prove the upper bound, third, we prove that the upper bound is tight for uniform pure ensembles. Finally we show the relation between Chernoff exponent and Wasserstein distance for uniform pure-state ensembles.

\paragraph{\textbf{Lower bound.}} The lower bound on \(E\) is obtained by first discriminating the finer hypotheses \((h,\alpha)\), which specify both the ensemble hypothesis and its labeling, and then merging all outcomes corresponding to the same \(h\). For \(h\in\{0,1\}\) and \(\alpha\in S_N\), define
\begin{equation}
\widetilde\rho_{h,\alpha}:=\bigoplus_{i=1}^Np_{\alpha(i)}^{(h)}\rho_{\alpha(i)}^{(h)},\qquad \mathcal C^{(h)}:=\{\widetilde\rho_{h,\alpha}:\alpha\in S_N\},\qquad \mathcal W:=\mathcal C^{(0)}\cup\mathcal C^{(1)}.
\label{eq:permuted-direct-sum-states}
\end{equation}
To relate this refined description to the original label-invariant experiment, define the permutation-averaged states
\begin{equation}
\overline\rho_h^{(M)}:=\frac1{|S_N|}\sum_{\alpha\in S_N}\widetilde\rho_{h,\alpha}^{\otimes M}.
\label{eq:permutation-averaged-states}
\end{equation}

For a direct-sum label sequence $\boldsymbol{\lambda}\in[N]^M$ with equality pattern $\Pi$ and $s=|\Pi|$ distinct labels, each injective assignment of these labels to intrinsic component labels extends to $(N-s)!$ permutations. Hence, the block indexed by $\boldsymbol{\lambda}$ equals $\Omega_{h,\Pi}/(N)_s$. Since there are $(N)_s$ observed-label sequences with equality pattern $\Pi$, index these identical orthogonal blocks by $a\in[(N)_s]$. Hence, up to a reordering of orthogonal direct-sum blocks,
\begin{equation}
\overline\rho_h^{(M)}\cong\bigoplus_{\Pi\in\Part([M])}\bigoplus_{a=1}^{(N)_s}\frac{\Omega_{h,\Pi}}{(N)_s},\qquad (N)_s:=\frac{N!}{(N-s)!}.
\label{eq:permutation-average-branch-decomposition}
\end{equation}
The additivity of the trace norm over direct sums and Theorem~\ref{thm:opt} therefore give
\begin{equation}
P_{e,\mathrm{opt}}^{(M)}=\frac12\left[1-\left\|\eta_0\overline\rho_0^{(M)}-\eta_1\overline\rho_1^{(M)}\right\|_1\right].
\label{eq:twirled-Bayes-error}
\end{equation}

Apply the multiple quantum Chernoff theorem \cite{Li2016multiplequantumchernoff10.1214/16-AOS1436}\footnote{It is required that the states have a common finite-dimensional joint support, which is satisfied since we only consider ensembles with finite states.} to the states \(\{\widetilde\rho_{h,\alpha}^{\otimes M}\}_{h,\alpha}\) with priors \(\eta_h/|S_N|\). There exists a sequence of POVMs \(\{\Lambda_{h,\alpha}^{(M)}\}_{h,\alpha}\) whose average error satisfies
\begin{equation}
P_{e,\mathrm{mult}}^{(M)}(\Lambda)\le\exp[-M(\xi_{\mathrm{mult}}-o(1))],
\label{eq:multiple-Chernoff-achievability}
\end{equation}
where
\[
\xi_{\mathrm{mult}}:=\min_{\substack{(h,\alpha)\ne(h',\beta)}}\sup_{0\le s\le1}\left\{-\log\Tr\!\left[(\widetilde\rho_{h,\alpha})^s(\widetilde\rho_{h',\beta})^{1-s}\right]\right\}.
\]
The direct-sum structure gives
\[
\Tr\!\left[(\widetilde\rho_{h,\alpha})^s(\widetilde\rho_{h',\beta})^{1-s}\right]=\sum_{i=1}^N\left(p_{\alpha(i)}^{(h)}\right)^s\left(p_{\beta(i)}^{(h')}\right)^{1-s}\Tr\!\left[\left(\rho_{\alpha(i)}^{(h)}\right)^s\left(\rho_{\beta(i)}^{(h')}\right)^{1-s}\right]=C_{\beta\circ\alpha^{-1}}^{h,h'}(s).
\]
Pairs with \(h\ne h'\) yield the across-hypothesis bottleneck \(\xi_{\mathrm{acr}}\), whereas pairs with \(h=h'\) and \(\alpha\ne\beta\) yield the in-hypothesis bottleneck \(\xi_{\mathrm{in}}\). Therefore,
\begin{equation}
\xi_{\mathrm{mult}}=\min\{\xi_{\mathrm{in}},\xi_{\mathrm{acr}}\}.
\label{eq:multiple-Chernoff-bottleneck}
\end{equation}

Merge the fine-grained outcomes according to the ensemble hypothesis by defining
\[
F_h^{(M)}:=\sum_{\alpha\in S_N}\Lambda_{h,\alpha}^{(M)},\qquad h=0,1.
\]
Then \(\{F_0^{(M)},F_1^{(M)}\}\) is a binary POVM. For every \(h\) and \(\alpha\),
\[
F_{1-h}^{(M)}=\sum_{\beta\in S_N}\Lambda_{1-h,\beta}^{(M)}\preceq I-\Lambda_{h,\alpha}^{(M)},
\]
and hence
\[
\sum_{h=0}^1\eta_h\Tr\!\left[\overline\rho_h^{(M)}F_{1-h}^{(M)}\right]=\sum_{h=0}^1\frac{\eta_h}{|S_N|}\sum_{\alpha\in S_N}\Tr\!\left[\widetilde\rho_{h,\alpha}^{\otimes M}F_{1-h}^{(M)}\right]\le P_{e,\mathrm{mult}}^{(M)}(\Lambda).
\]
Using Eq.~\eqref{eq:twirled-Bayes-error} and Eq.~\eqref{eq:multiple-Chernoff-achievability}, we obtain
\[
P_{e,\mathrm{opt}}^{(M)}\le\exp[-M(\min\{\xi_{\mathrm{in}},\xi_{\mathrm{acr}}\}-o(1))],
\]
and therefore
\begin{equation}
E\ge\min\{\xi_{\mathrm{in}},\xi_{\mathrm{acr}}\}.
\label{eq:Chernoff-exponent-lower-bound}
\end{equation}

\paragraph{\textbf{Upper bound.}} Fix \(\pi\in S_N\), set \(\widetilde\rho_0:=\widetilde\rho_{0,\mathrm{id}}\), and let \(P_{e,\pi,\mathrm{ref}}^{(M)}\) be the optimal Bayes error for the refined i.i.d. test
\[
\widetilde\rho_0^{\otimes M}\quad\text{versus}\quad\widetilde\rho_{1,\pi}^{\otimes M}.
\]
Coarse-graining the direct-sum index sequence to its equality pattern recovers the original label-invariant experiment. The refined test therefore has no larger error:
\[
P_{e,\mathrm{opt}}^{(M)}\ge P_{e,\pi,\mathrm{ref}}^{(M)}.
\]
The quantum Chernoff theorem \cite{audenaert} gives
\[
\lim_{M\to\infty}-\frac1M\log P_{e,\pi,\mathrm{ref}}^{(M)}=\sup_{0\le s\le1}\bigl[-\log C_\pi^{0,1}(s)\bigr]=\xi_\pi^{0,1}.
\]
Consequently, \(E\le\xi_\pi^{0,1}\) for every \(\pi\in S_N\), and hence
\begin{equation}
E\le\min_{\pi\in S_N}\xi_\pi^{0,1}=\xi_{\mathrm{acr}}.
\label{eq:Chernoff-exponent-upper-bound}
\end{equation}

\paragraph{\textbf{Achievability of upper bound for uniform pure ensembles}}
To show the achievability, we construct a specific POVM for any pair of uniform pure ensembles and directly give the exact Bayes error $P^{(M)}_e\le c \exp{-M\xi_\text{acr}}$, hence the corresponding Chernoff exponent $E\ge\xi_\text{acr}$, since we already have $E\le\xi_\text{acr}$, we then get $E=\xi_\text{acr}$.

Assume $p_i^{(h)}=1/N$ and $\rho_i^{(h)}=\ketbra{\psi_i^{(h)}}$. For a possible branch $\Pi\in\Part([M])$, set $\mathcal A_\Pi:=\operatorname{Inj}(\Pi,[N])$, representing all possible injections indicating the intrinsic states of labels in the observed parition. For $\alpha\in\mathcal A_\Pi$, define $\ket{\Phi_{\Pi,\alpha}^{(h)}}:=\bigotimes_{B\in\Pi}\bigotimes_{m\in B}\ket{\psi_{\alpha(B)}^{(h)}}_m$, where the tensor factors are placed in the physical sample order. Then the state representing the partiion $\Pi$ can also be written as
\begin{equation}
\Omega_{h,\Pi}=N^{-M}\sum_{\alpha\in\mathcal A_\Pi}\ketbra{\Phi_{\Pi,\alpha}^{(h)}}.
\label{eq:pure-uniform-branch-expansion}
\end{equation}

We now give the specific POVM used for testing. For each branch $\Pi$, use the support test $E_{1|\Pi}:=P_{1,\Pi}$ and $E_{0|\Pi}:=I-P_{1,\Pi}$. These branchwise measurements form a valid block-diagonal POVM. Since $\Omega_{1,\Pi}$ is supported on $P_{1,\Pi}$, the error under $H_1$ vanishes, and hence the exact Bayes error is
\begin{equation}
P_e^{(M)}\le\eta_0\sum_{\Pi\in\Part([M])}\Tr(\Omega_{0,\Pi}P_{1,\Pi}).
\end{equation}

To show $P^{(M)}_e\le c \exp{-M\xi_\text{acr}}$, we first prove a lemma.

\begin{lemma}
    Define $T_{1,\Pi}:=\sum_{\alpha\in\mathcal A_\Pi}\ketbra{\Phi_{\Pi,\alpha}^{(1)}}$ and $P_{1,\Pi}:=\operatorname{proj}\!\left(\Span\{\ket{\Phi_{\Pi,\alpha}^{(1)}}:\alpha\in\mathcal A_\Pi\}\right)$. We claim that there exists a constant $\kappa>0$, depending only on $\mathcal E_1$, such that
\begin{equation}
T_{1,\Pi}\succeq\kappa P_{1,\Pi}
\label{eq:uniform-branch-frame-bound}
\end{equation}
for every $M$ and every possible branch $\Pi$.
\end{lemma}

\begin{proof}
    Let $V_\Pi:\mathbb C^{\mathcal A_\Pi}\to\mathcal H^{\otimes M}$ be the synthesis map defined by $V_\Pi e_\alpha=\ket{\Phi_{\Pi,\alpha}^{(1)}}$, where $\{e_\alpha\}_{\alpha\in\mathcal A_\Pi}$ is the standard basis of $\mathbb C^{\mathcal A_\Pi}$\footnote{Equivalently, $V_\Pi$ is the matrix whose columns are the branch vectors, and $\dim \mathbb C^{\mathcal A_\Pi}$ is the number of possible injections $\frac{N!}{(N-|\Pi|)!}$.}. Since $T_{1,\Pi}=V_\Pi V_\Pi^\dagger$ and $P_{1,\Pi}$ is the projector onto $\operatorname{ran}V_\Pi$, it suffices to prove that the nonzero singular values of all $V_\Pi$ are bounded below by a positive constant independent of $M$ and $\Pi$.

Suppose otherwise. Then there exist integers $M_\ell$, branches $\Pi_\ell\in\Part([M_\ell])$, and unit vectors $z_\ell\in(\ker V_{\Pi_\ell})^\perp$ such that $\norm{V_{\Pi_\ell}z_\ell}\to0$. Since $|\Pi_\ell|\le N$, after passing to a subsequence we may assume $|\Pi_\ell|=s$ for all $\ell$. Order the blocks as $\Pi_\ell=\{B_{1,\ell},\ldots,B_{s,\ell}\}$ and identify each $\mathcal A_{\Pi_\ell}$ with the fixed set $\mathcal A_s:=\operatorname{Inj}([s],[N])$. After permuting the output registers, which preserves norms, singular values, and kernels, write the resulting synthesis map as $V_\ell$. Its columns are $V_\ell e_\alpha=\bigotimes_{a=1}^s\ket{\psi_{\alpha(a)}^{(1)}}^{\otimes n_{a,\ell}}$, where $n_{a,\ell}:=|B_{a,\ell}|$. Since $\norm{V_\ell}\le\sqrt{|\mathcal A_s|}$, compactness of the unit sphere in $\mathbb C^{\mathcal A_s}$ gives a further subsequence such that $z_\ell\to z_\infty$ with $\norm{z_\infty}=1$. Writing $z_\infty=((z_\infty)_\alpha)_{\alpha\in\mathcal A_s}$ in the standard basis of $\mathbb C^{\mathcal A_s}$, we have
\begin{equation}
\norm{V_\ell z_\infty}\le\norm{V_\ell z_\ell}+\norm{V_\ell}\norm{z_\infty-z_\ell}\longrightarrow0.
\end{equation}

Passing to another subsequence, for each $a\in[s]$, either $n_{a,\ell}$ is eventually constant or $n_{a,\ell}\to\infty$. Let $\mathsf F$ and $\mathsf L$ denote the corresponding sets of indices, and for $a\in\mathsf F$ let $n_a$ be the eventual value of $n_{a,\ell}$. Define an equivalence relation on $[N]$ by $i\sim j$ if and only if $\rho_i^{(1)}=\rho_j^{(1)}$. Since the set of distinct states in an ensemble is finite, there exists $q<1$ such that $|\braket{\psi_i^{(1)}}{\psi_j^{(1)}}|\le q$ whenever $i\not\sim j$.

Group the elements of $\mathcal A_s$ into equivalence classes as follows: $\alpha$ and $\beta$ belong to the same equivalence class $C$ if and only if $\alpha(a)$ and $\beta(a)$ correspond to the same state (or $\alpha(a)\sim\beta(a)$) for every $a\in\mathsf L$. For each class $C$, choose $\alpha_C\in C$ and define $u_{C,\ell}:=\bigotimes_{a\in\mathsf L}\ket{\psi_{\alpha_C(a)}^{(1)}}^{\otimes n_{a,\ell}}$, $w_\alpha:=\bigotimes_{a\in\mathsf F}\ket{\psi_{\alpha(a)}^{(1)}}^{\otimes n_a}$, and $y_C:=\sum_{\alpha\in C}(z_\infty)_\alpha w_\alpha$. Then for all sufficiently large $\ell$ we have $V_\ell z_\infty=\sum_Cu_{C,\ell}\otimes y_C$.

If $C\ne C'$, then for some $a\in\mathsf L$ the labels $\alpha_C(a)$ and $\alpha_{C'}(a)$ correspond to distinct states. Since different classes become asymptotically orthogonal, expanding the squared norm gives
\[
\lim_{\ell\to\infty}\|V_\ell z_\infty\|^2=\sum_C\|y_C\|^2.
\]
The left-hand side is zero, and the right-hand side is a finite sum of nonnegative constants independent of $\ell$. Hence $y_C=0$ for every class $C$. For sufficiently large $\ell$, the decomposition $V_\ell z_\infty=\sum_Cu_{C,\ell}\otimes y_C$ holds exactly, and therefore $V_\ell z_\infty=0$, that is,  $z_\infty\in\ker V_\ell$. Since $z_\ell\in(\ker V_\ell)^\perp$, we have $\langle z_\ell,z_\infty\rangle=0$ for all sufficiently large $\ell$, contradicting $z_\ell\to z_\infty$ and $\norm{z_\infty}=1$. Thus there exists $\gamma>0$ such that every nonzero singular value of every $V_\Pi$ is at least $\gamma$. Consequently, $T_{1,\Pi}=V_\Pi V_\Pi^\dagger\succeq\gamma^2P_{1,\Pi}$, proving Eq.~\eqref{eq:uniform-branch-frame-bound} with $\kappa=\gamma^2$.
\end{proof}

Define
\begin{equation}
F_{ij}:=\Tr(\rho_i^{(0)}\rho_j^{(1)})=|\braket{\psi_i^{(0)}}{\psi_j^{(1)}}|^2,\qquad F_\star:=\max_{\pi\in S_N}\frac1N\sum_{i=1}^NF_{i,\pi(i)}.
\label{eq:F-star-definition}
\end{equation}
For a rank-one projector, with the standard support convention at exponent zero, $(\rho_i^{(h)})^s=\rho_i^{(h)}$ for every $0\le s\le1$. Therefore,
\begin{equation}
\xi_\pi^{0,1}(s)=-\log\!\left(\frac1N\sum_{i=1}^NF_{i,\pi(i)}\right),\qquad \xi_{\mathrm{acr}}=-\log F_\star.
\label{eq:pure-uniform-xi-acr}
\end{equation}

Using $P_{1,\Pi}\preceq\kappa^{-1}T_{1,\Pi}$ and Eq.~\eqref{eq:pure-uniform-branch-expansion}, we obtain
\begin{equation}
P_e^{(M)}\le\eta_0\kappa^{-1}N^{-M}\sum_{\Pi}\sum_{\alpha,\beta\in\mathcal A_\Pi}|\braket{\Phi_{\Pi,\alpha}^{(0)}}{\Phi_{\Pi,\beta}^{(1)}}|^2.
\label{eq:support-test-branch-bound}
\end{equation}
For fixed $\Pi$, $\alpha$, and $\beta$, the overlap factorizes as $|\braket{\Phi_{\Pi,\alpha}^{(0)}}{\Phi_{\Pi,\beta}^{(1)}}|^2=\prod_{B\in\Pi}F_{\alpha(B),\beta(B)}^{|B|}$. Equivalently, assigning $i_m=\alpha(B)$ and $j_m=\beta(B)$ for the unique block $B\in\Pi$ containing $m$ rewrites Eq.~\eqref{eq:support-test-branch-bound} as
\begin{equation}
P_e^{(M)}\le\eta_0\kappa^{-1}N^{-M}\sum_{\substack{\mathbf i,\mathbf j\in[N]^M\\\eqp(\mathbf i)=\eqp(\mathbf j)}}\prod_{m=1}^MF_{i_m,j_m}.
\label{eq:support-test-sequence-bound}
\end{equation}

Whenever $\eqp(\mathbf i)=\eqp(\mathbf j)$, the correspondence sending each label appearing in $\mathbf i$ to the label occupying the same positions in $\mathbf j$ is a bijection between the two used-label sets. It therefore extends to a permutation $\pi\in S_N$ satisfying $j_m=\pi(i_m)$ for every $m$. Summing over all $\pi\in S_N$ may count a pair $(\mathbf i,\mathbf j)$ more than once \footnote{For example, the indices permutated may not appear in all $M$ samples.}, but all summands are nonnegative. Thus
\begin{equation}
P_e^{(M)}\le\eta_0\kappa^{-1}N^{-M}\sum_{\pi\in S_N}\sum_{\mathbf i\in[N]^M}\prod_{m=1}^MF_{i_m,\pi(i_m)}.
\end{equation}
For each fixed permutation $\pi$, the inner sum factorizes:
\begin{equation}
N^{-M}\sum_{\mathbf i\in[N]^M}\prod_{m=1}^MF_{i_m,\pi(i_m)}=\left(\frac1N\sum_{i=1}^NF_{i,\pi(i)}\right)^M.
\end{equation}
Consequently,
\begin{equation}
P_e^{(M)}\le\eta_0\kappa^{-1}\sum_{\pi\in S_N}\left(\frac1N\sum_{i=1}^NF_{i,\pi(i)}\right)^M\le\eta_0\kappa^{-1}N!F_\star^M.
\label{eq:pure-uniform-achievability}
\end{equation}

If $F_\star>0$, Eq.~\eqref{eq:pure-uniform-achievability} implies
\begin{equation}
E=\lim_{M\to\infty}-\frac1M\log P_e^{(M)}\ge-\log F_\star=\xi_{\mathrm{acr}}.
\end{equation}
If $F_\star=0$, Eq.~\eqref{eq:pure-uniform-achievability} gives $P_{e,\mathrm{opt}}^{(M)}=0$ for every $M\ge1$, so the same conclusion holds with the convention $-\log0=+\infty$. Combining this achievability bound with the preceding converse $E\le\xi_{\mathrm{acr}}$ proves that
\begin{equation}
E=\xi_{\mathrm{acr}}.
\end{equation}

\paragraph{\textbf{Relation between the Chernoff exponent and the Wasserstein distance.}} Recall \(F_{ij}:=\Tr(\rho_i^{(0)}\rho_j^{(1)})\) and \(C_{ij}:=1-F_{ij}\). For two uniform ensembles of equal cardinality \(N\), define
\begin{equation}
\Gamma_N:=\left\{\gamma\in\mathbb R_+^{N\times N}:\sum_{j=1}^N\gamma_{ij}=\frac1N,\ \sum_{i=1}^N\gamma_{ij}=\frac1N\right\},
\label{eq:uniform-transport-plans}
\end{equation}
and the Wasserstein distance associated with the cost \(C_{ij}\) by
\begin{equation}
W(\calE_0,\calE_1):=\min_{\gamma\in\Gamma_N}\sum_{i,j=1}^N\gamma_{ij}C_{ij}.
\label{eq:uniform-wasserstein-distance}
\end{equation}

For every \(\gamma\in\Gamma_N\), the matrix \(N\gamma\) is doubly stochastic. By the Birkhoff--von Neumann theorem, there exist coefficients \(\lambda_\pi\ge0\), \(\pi\in S_N\), with \(\sum_{\pi\in S_N}\lambda_\pi=1\), such that
\[
\gamma_{ij}=\frac1N\sum_{\pi\in S_N}\lambda_\pi\delta_{j,\pi(i)}.
\]
Since the transport cost is linear in \(\gamma\),
\[
\sum_{i,j=1}^N\gamma_{ij}C_{ij}=\sum_{\pi\in S_N}\lambda_\pi\frac1N\sum_{i=1}^NC_{i,\pi(i)}\ge\min_{\pi\in S_N}\frac1N\sum_{i=1}^NC_{i,\pi(i)}.
\]
Conversely, for every \(\pi\in S_N\), the permutation plan \(\gamma_{ij}=N^{-1}\delta_{j,\pi(i)}\) belongs to \(\Gamma_N\). Therefore,
\begin{align}
W(\calE_0,\calE_1)&=\min_{\pi\in S_N}\frac1N\sum_{i=1}^NC_{i,\pi(i)}=\min_{\pi\in S_N}\frac1N\sum_{i=1}^N\left(1-F_{i,\pi(i)}\right)\notag\\
&=1-\max_{\pi\in S_N}\frac1N\sum_{i=1}^NF_{i,\pi(i)}=1-F_\star.
\label{eq:wasserstein-matching-fidelity}
\end{align}
Since the preceding result gives \(E=\xi_{\mathrm{acr}}=-\log F_\star\) for finite uniform pure-state ensembles, Eq.~\eqref{eq:wasserstein-matching-fidelity} yields
\begin{equation}
E=\xi_{\mathrm{acr}}=-\log\!\left[1-W(\calE_0,\calE_1)\right],
\label{eq:chernoff-wasserstein-relation}
\end{equation}
or equivalently \(W(\calE_0,\calE_1)=1-e^{-E}=1-e^{-\xi_{\mathrm{acr}}}\). This completes the proof.
\end{proof}

\section{Specific calculations and proofs in examples}
\label{app:extra_proofs}

\subsection{Proofs in testing state design}

\paragraph{\textbf{Haar-random pure states and state designs.}} Let \(\mathcal H\simeq\mathbb C^d\), and let \(d\psi\) denote the normalized unitarily invariant measure on the pure states of \(\mathcal H\). Equivalently, a Haar-random pure state can be generated as \(\ket{\psi}=U\ket{\psi_0}\), where \(U\) is Haar distributed on \(U(d)\) and \(\ket{\psi_0}\) is any fixed reference state. The induced distribution is independent of \(\ket{\psi_0}\). Writing \(\rho_\psi:=\ket{\psi}\!\bra{\psi}\), its unitary invariance means that
\[
\int d\psi\,f(\rho_\psi)=\int d\psi\,f(U\rho_\psi U^\dagger)
\]
for every fixed \(U\in U(d)\) and every integrable function \(f\).

The \(r\)-th moment operator of the Haar ensemble is
\begin{equation}
M_{\mathrm{Haar}}^{(r)}:=\int d\psi\,\rho_\psi^{\otimes r}=\frac{\Pi_{\mathrm{sym}}^{(r)}}{D_r},\qquad D_r:=\dim\operatorname{Sym}^r(\mathcal H)=\binom{d+r-1}{r}.
\label{eq:Haar-moment}
\end{equation}
Indeed, every \(\rho_\psi^{\otimes r}\) is supported on the symmetric subspace of \(\mathcal H^{\otimes r}\), while the Haar average is invariant under \(U^{\otimes r}\) for every \(U\in U(d)\); Schur's lemma and unit trace then give Eq.~\eqref{eq:Haar-moment} \cite{Harrow2013}. Here
\begin{equation}
\Pi_{\mathrm{sym}}^{(r)}:=\frac1{r!}\sum_{\pi\in S_r}U_\pi
\label{eq:symmetric-projector}
\end{equation}
is the projector onto the symmetric subspace, and \(U_\pi\) permutes the \(r\) tensor factors according to \(\pi\). In particular, \(M_{\mathrm{Haar}}^{(1)}=I/d\) and \(\Tr M_{\mathrm{Haar}}^{(r)}=1\).

Consider a finite weighted pure-state ensemble \(\mathcal E=\{(p_i,\rho_i)\}_{i=1}^N\), where \(\rho_i=\ket{\psi_i}\!\bra{\psi_i}\), \(p_i>0\), and \(\sum_{i=1}^Np_i=1\). Its \(r\)-th ordinary moment operator is
\begin{equation}
M_{\mathcal E}^{(r)}:=\sum_{i=1}^Np_i\rho_i^{\otimes r}.
\label{eq:ordinary-ensemble-moment}
\end{equation}
The ensemble \(\mathcal E\) is an exact state \(t\)-design, equivalently a weighted complex projective \(t\)-design, if its \(t\)-th moment agrees with the Haar moment \cite{ambainis2007quantum}:
\begin{equation}
M_{\mathcal E}^{(t)}=M_{\mathrm{Haar}}^{(t)}=\frac{\Pi_{\mathrm{sym}}^{(t)}}{D_t},\qquad D_t=\binom{d+t-1}{t}.
\label{eq:state-t-design}
\end{equation}
This equality implies moment matching at every lower order. Indeed, for \(1\le r<t\),
\[
\Tr_{r+1,\ldots,t}M_{\mathcal E}^{(t)}=M_{\mathcal E}^{(r)},\qquad \Tr_{r+1,\ldots,t}M_{\mathrm{Haar}}^{(t)}=M_{\mathrm{Haar}}^{(r)},
\]
and therefore
\begin{equation}
M_{\mathcal E}^{(r)}=M_{\mathrm{Haar}}^{(r)}=\frac{\Pi_{\mathrm{sym}}^{(r)}}{D_r},\qquad 1\le r\le t.
\label{eq:lower-design-moments}
\end{equation}
Thus, every state \(t\)-design is automatically a state \(r\)-design for each \(r\le t\), while no corresponding constraint is imposed on moments of order \(r>t\).

Moment matching also has an operational formulation. For every \(1\le r\le t\) and every operator \(A\) on \(\mathcal H^{\otimes r}\),
\begin{equation}
\sum_{i=1}^Np_i\Tr(A\rho_i^{\otimes r})=\int d\psi\,\Tr(A\rho_\psi^{\otimes r}).
\label{eq:operational-design-matching}
\end{equation}
Consequently, every ensemble-averaged measurement statistic obtained from at most \(t\) identical copies of a sampled state agrees with the corresponding Haar average. Equivalently, a state \(t\)-design reproduces Haar averages of all phase-invariant homogeneous polynomials of bidegree \((r,r)\) in the state amplitudes and their complex conjugates for every \(r\le t\) \cite{AmbainisEmerson2007}.

An equivalent characterization is provided by the \(t\)-th frame potential
\begin{equation}
\mathcal F_t(\mathcal E):=\sum_{i,j=1}^Np_ip_j|\langle\psi_i|\psi_j\rangle|^{2t}=\Tr\!\left[\left(M_{\mathcal E}^{(t)}\right)^2\right].
\label{eq:frame-potential}
\end{equation}
Since \(M_{\mathcal E}^{(t)}\) is a unit-trace positive operator supported on the \(D_t\)-dimensional symmetric subspace,
\begin{equation}
\mathcal F_t(\mathcal E)\ge\frac1{D_t}=\binom{d+t-1}{t}^{-1}.
\label{eq:frame-potential-bound}
\end{equation}
Equality holds if and only if \(M_{\mathcal E}^{(t)}=\Pi_{\mathrm{sym}}^{(t)}/D_t\), and hence if and only if \(\mathcal E\) is an exact state \(t\)-design.

\paragraph{\textbf{Proof of Corollary~\ref{cor:first-informative-moment}.}} For a uniform ensemble of cardinality \(N\), the power-weighted and ordinary moments satisfy \(\mathcal K_h^{(r)}=N^{1-r}M_h^{(r)}\). Hence, Eq.~\eqref{eq:Mobius-uniform} gives, for every \(\Pi\in\Part([M])\),
\[
\Omega_{h,\Pi}=\sum_{R\succeq\Pi}\mu(\Pi,R)N^{|R|-M}\bigotimes_{D\in R}M_h^{(|D|)},
\]
where each moment operator acts on the registers indexed by \(D\). If \(M<k\), every block \(D\) of every \(R\in\Part([M])\) satisfies \(|D|\le M<k\). Since \(M_0^{(r)}=M_1^{(r)}\) for \(1\le r<k\), all terms in the M\"obius expansions agree under the two hypotheses, so \(\Omega_{0,\Pi}=\Omega_{1,\Pi}\) for every \(\Pi\in\Part([M])\). The equal-prior Bayes formula in Eq.~\eqref{eq:Bayes-exact} therefore gives \(P_{e,\mathrm{opt}}^{(M)}=1/2\).

Now let \(M=k\), and denote the one-block partition of \([k]\) by \(\widehat 1_k:=\{[k]\}\). If \(R\neq\widehat 1_k\), then \(R\) contains at least two blocks, so \(|D|<k\) for every \(D\in R\); all such terms again cancel between the two hypotheses. Thus, only \(R=\widehat 1_k\) contributes to the branch-state difference:
\begin{equation}
\Omega_{0,\Pi}-\Omega_{1,\Pi}=\mu(\Pi,\widehat 1_k)N^{1-k}\left(M_0^{(k)}-M_1^{(k)}\right).
\label{eq:first-informative-branch-difference}
\end{equation}
Since \(\widehat 1_k\) is obtained by merging all \(|\Pi|\) blocks of \(\Pi\), \(\mu(\Pi,\widehat 1_k)=(-1)^{|\Pi|-1}(|\Pi|-1)!\). Taking the trace norm in Eq.~\eqref{eq:first-informative-branch-difference} gives
\begin{equation}
\left\|\Omega_{0,\Pi}-\Omega_{1,\Pi}\right\|_1=(|\Pi|-1)!N^{1-k}\left\|M_0^{(k)}-M_1^{(k)}\right\|_1.
\label{eq:first-informative-branch-norm}
\end{equation}
There are \(\left\{{k\atop\ell}\right\}\) partitions of \([k]\) with \(\ell\) blocks. Summing Eq.~\eqref{eq:first-informative-branch-norm} over \(\Pi\in\Part([k])\) yields
\begin{equation}
\sum_{\Pi\in\Part([k])}\left\|\Omega_{0,\Pi}-\Omega_{1,\Pi}\right\|_1=A_kN^{1-k}\left\|M_0^{(k)}-M_1^{(k)}\right\|_1,\qquad A_k:=\sum_{\ell=1}^k\left\{{k\atop\ell}\right\}(\ell-1)!.
\label{eq:first-informative-total-norm}
\end{equation}
Substituting Eq.~\eqref{eq:first-informative-total-norm} into the equal-prior Bayes formula gives
\[
P_{e,\mathrm{opt}}^{(k)}=\frac12\left[1-\frac12A_kN^{1-k}\left\|M_0^{(k)}-M_1^{(k)}\right\|_1\right].
\]
This proves the corollary. \hfill\(\square\)

\paragraph{\textbf{Proof of Theorem~\ref{thm: unversal chernoff t-design upper bound}.}} Write \(\rho_i^{(h)}=\ket{\psi_i^{(h)}}\!\bra{\psi_i^{(h)}}\), and define the cross fidelities and optimal matching fidelity by
\begin{equation}
F_{ij}:=\Tr\!\left(\rho_i^{(0)}\rho_j^{(1)}\right)=\left|\left\langle\psi_i^{(0)}\middle|\psi_j^{(1)}\right\rangle\right|^2,\qquad F_\pi:=\frac1N\sum_{i=1}^NF_{i,\pi(i)},\qquad F_\star:=\max_{\pi\in S_N}F_\pi.
\label{eq:design-optimal-matching-fidelity}
\end{equation}
Since the ensembles are finite, uniform, and pure, Theorem~\ref{thm:fixed-N-chernoff-bounds} gives \(E=-\log F_\star\).

Let \(D_r:=\binom{d+r-1}{r}\), with \(D_0:=1\). For every fixed \(i,j\) and \(0\le r\le t\), the \(t\)-design property gives
\begin{equation}
\frac1N\sum_{j=1}^NF_{ij}^r=\Tr\!\left[\left(\rho_i^{(0)}\right)^{\otimes r}\frac1N\sum_{j=1}^N\left(\rho_j^{(1)}\right)^{\otimes r}\right]=\frac1{D_r},\qquad \frac1N\sum_{i=1}^NF_{ij}^r=\frac1{D_r},
\label{eq:t-design-overlap-moments}
\end{equation}
where the identities for \(r=0\) are immediate. Define the Haar overlap measure
\begin{equation}
\mu_d(dx):=(d-1)(1-x)^{d-2}dx,\qquad x\in[0,1],
\label{eq:Haar-overlap-measure}
\end{equation}
whose moments satisfy \(\int_0^1x^r\mu_d(dx)=1/D_r\).

Let \(\mathcal P_{t-1}^{+}\) be the set of nonzero real polynomials of degree at most \(t-1\) that are nonnegative on \([0,1]\). For \(p(x)=\sum_{r=0}^{t-1}a_rx^r\in\mathcal P_{t-1}^{+}\), define
\begin{equation}
A_p:=\sum_{r=0}^{t-1}\frac{a_r}{D_r}=\int_0^1p(x)\mu_d(dx)>0,\qquad R[p]:=\frac{\sum_{r=0}^{t-1}a_r/D_{r+1}}{\sum_{r=0}^{t-1}a_r/D_r}=\frac{\int_0^1xp(x)\mu_d(dx)}{\int_0^1p(x)\mu_d(dx)}.
\label{eq:design-polynomial-ratio}
\end{equation}
The strict positivity of \(A_p\) follows because \(\mu_d\) has full support on \([0,1]\) and \(p\) is nonzero and nonnegative. Define
\begin{equation}
B_{ij}:=\frac{p(F_{ij})}{NA_p}.
\label{eq:design-doubly-stochastic-matrix}
\end{equation}
Using Eq.~\eqref{eq:t-design-overlap-moments},
\[
\sum_{j=1}^NB_{ij}=\frac1{NA_p}\sum_{r=0}^{t-1}a_r\sum_{j=1}^NF_{ij}^r=\frac1{A_p}\sum_{r=0}^{t-1}\frac{a_r}{D_r}=1,
\]
and the same calculation gives \(\sum_{i=1}^NB_{ij}=1\). Thus \(B\) is doubly stochastic. By the Birkhoff--von Neumann theorem, there exist \(\lambda_\pi\ge0\), \(\pi\in S_N\), with \(\sum_{\pi\in S_N}\lambda_\pi=1\), such that \(B=\sum_{\pi\in S_N}\lambda_\pi P_\pi\), where \((P_\pi)_{ij}=\delta_{j,\pi(i)}\). Therefore,
\begin{align}
\sum_{\pi\in S_N}\lambda_\pi F_\pi&=\frac1N\sum_{i,j=1}^NB_{ij}F_{ij}=\frac1{N^2A_p}\sum_{i,j=1}^NF_{ij}p(F_{ij})\notag\\
&=\frac1{N^2A_p}\sum_{r=0}^{t-1}a_r\sum_{i,j=1}^NF_{ij}^{r+1}=\frac1{A_p}\sum_{r=0}^{t-1}\frac{a_r}{D_{r+1}}=R[p].
\label{eq:design-Birkhoff-polynomial-average}
\end{align}
Since the left-hand side is a convex combination of \(\{F_\pi\}_{\pi\in S_N}\), \(F_\star\ge R[p]\). Optimizing over \(p\) gives
\begin{equation}
F_\star\ge q_{d,t}:=\sup_{p\in\mathcal P_{t-1}^{+}}R[p],\qquad E=-\log F_\star\le-\log q_{d,t}.
\label{eq:design-polynomial-matching-bound}
\end{equation}

The feasible polynomial \(p(x)=x^{t-1}\) gives
\[
R[p]=\frac{1/D_t}{1/D_{t-1}}=\frac{D_{t-1}}{D_t}=\frac{t}{t+d-1},
\]
and hence
\begin{equation}
E\le-\log q_{d,t}\le\log\left(1+\frac{d-1}{t}\right).
\label{eq:elementary-design-exponent-bound}
\end{equation}

It remains to determine \(q_{d,t}\). First let \(t=2n-1\), so \(\deg p\le2n-2\). The \(n\)-point Gauss--Jacobi quadrature rule for \(\mu_d\), with nodes \(0<x_1<\cdots<x_n<1\) and positive weights \(\omega_1,\ldots,\omega_n\), is exact for polynomials of degree at most \(2n-1\). Since \(\deg p\le2n-2\) and \(\deg(xp)\le2n-1\), Eq.~\eqref{eq:design-polynomial-ratio} becomes
\[
R[p]=\frac{\sum_{k=1}^n\omega_kx_kp(x_k)}{\sum_{k=1}^n\omega_kp(x_k)}\le x_n,
\]
because the ratio is a convex combination of the nodes. Equality is attained by
\[
p^\star(x):=c\prod_{k=1}^{n-1}(x-x_k)^2,\qquad c>0,
\]
which belongs to \(\mathcal P_{2n-2}^{+}\), vanishes at \(x_1,\ldots,x_{n-1}\), and is positive at \(x_n\). Therefore,
\begin{equation}
q_{d,2n-1}=x_n.
\label{eq:odd-design-polynomial-optimum}
\end{equation}

Now let \(t=2n\), so \(\deg p\le2n-1\). The left Gauss--Radau quadrature rule for \(\mu_d\), with fixed node \(x_0=0\), free nodes \(0<x_1<\cdots<x_n<1\), and positive weights \(\omega_0,\ldots,\omega_n\), is exact for polynomials of degree at most \(2n\). Hence
\[
R[p]=\frac{\sum_{k=0}^n\omega_kx_kp(x_k)}{\sum_{k=0}^n\omega_kp(x_k)}\le x_n.
\]
Equality is attained by
\[
p^\star(x):=cx\prod_{k=1}^{n-1}(x-x_k)^2,\qquad c>0,
\]
which belongs to \(\mathcal P_{2n-1}^{+}\) and vanishes at every quadrature node except \(x_n\). Thus,
\begin{equation}
q_{d,2n}=x_n.
\label{eq:even-design-polynomial-optimum}
\end{equation}

The Gauss nodes for \(\mu_d(dx)\propto(1-x)^{d-2}dx\) are the zeros of the shifted Jacobi polynomial
\[
\widetilde P_n^{(d-2,0)}(x):=P_n^{(d-2,0)}(2x-1).
\]
For the Gauss--Radau rule, the free nodes are the Gauss nodes for the modified measure \(x\mu_d(dx)\propto x(1-x)^{d-2}dx\), and hence are the zeros of \(\widetilde P_n^{(d-2,1)}\). Defining
\[
n:=\left\lceil\frac t2\right\rceil,\qquad \beta_t:=\begin{cases}0,&t\text{ odd},\\1,&t\text{ even},\end{cases}
\]
Eqs.~\eqref{eq:odd-design-polynomial-optimum} and \eqref{eq:even-design-polynomial-optimum} yield
\begin{equation}
q_{d,t}=x_{\max}\!\left(\widetilde P_{\lceil t/2\rceil}^{(d-2,\beta_t)}\right)=:x_{\max}.
\label{eq:Jacobi-design-polynomial-optimum}
\end{equation}
Substituting Eq.~\eqref{eq:Jacobi-design-polynomial-optimum} into Eq.~\eqref{eq:design-polynomial-matching-bound} gives
\begin{equation}
E\le-\log x_{\max}.
\label{eq:Jacobi-design-exponent-bound}
\end{equation}

For completeness, the \(t^{-2}\) behavior follows from the extreme-zero asymptotics of Jacobi polynomials. Set \(\alpha:=d-2\) and write \(2x_{\max}-1=\cos\theta_{n,1}\), where \(\theta_{n,1}\) is the smallest positive zero of \(P_n^{(\alpha,\beta_t)}(\cos\theta)\). With $\varrho_t:=n+(\alpha+\beta_t+1)/2=(t+d)/2$, the Jacobi-zero expansion of Frenzen and Wong~\cite{FrenzenWong1985}, applied separately to the odd and even subsequences, gives $\theta_{n,1}=j_{\alpha,1}/\varrho_t+O_d(\varrho_t^{-3})$, where we used $1-\phi\cot\phi=O(\phi^2)$ and $\tan(\phi/2)=O(\phi)$ for $\phi=j_{\alpha,1}/\varrho_t$, and $j_{\alpha,1}$ is the first positive zero of the Bessel function $J_\alpha$. Consequently,
\begin{align}
1-x_{\max}&=\frac{1-\cos\theta_{n,1}}2=\frac{j_{d-2,1}^2}{(t+d)^2}+O_d(t^{-4}),\notag\\
-\log x_{\max}&=\frac{j_{d-2,1}^2}{(t+d)^2}+O_d(t^{-4})=O_d(t^{-2}).
\label{eq:Jacobi-design-root-asymptotics}
\end{align}
This completes the proof. \hfill\(\square\)

\paragraph{\textbf{Restatement of Corollary~3 and proof.}}

\begin{corollary}[Finite-sample converse for testing state designs]
Let \(\mathcal E_h=\{(1/N,\rho_i^{(h)})\}_{i=1}^N\), \(h=0,1\), be two finite uniform pure-state \(t\)-designs of equal cardinality on \(\mathbb C^d\), with priors \(\eta_0,\eta_1>0\) satisfying \(\eta_0+\eta_1=1\). Let \(x_{\max}\) be defined as in Theorem~\ref{thm: unversal chernoff t-design upper bound}. Then, for every \(M\ge1\),
\begin{equation}
P_{e,\mathrm{opt}}^{(M)}\ge\frac12\left[1-\sqrt{1-4\eta_0\eta_1x_{\max}^M}\right].
\label{eq:finite-sample-design-converse}
\end{equation}
Moreover, \(P_{e,\mathrm{opt}}^{(M)}=\min\{\eta_0,\eta_1\}\) for every \(M\le t\). For \(0<\epsilon<\min\{\eta_0,\eta_1\}\), define \(M_\epsilon:=\inf\{M\in\mathbb N:M\ge1,\ P_{e,\mathrm{opt}}^{(M)}\le\epsilon\}\), with \(\inf\varnothing:=+\infty\). Then
\begin{equation}
M_\epsilon\ge\max\left\{t+1,\left\lceil\frac{\log[\eta_0\eta_1/(\epsilon(1-\epsilon))]}{-\log x_{\max}}\right\rceil\right\}.
\label{eq:finite-sample-design-complexity-bound}
\end{equation}
Consequently, for fixed \(d\), \(\eta_0\), \(\eta_1\), and \(\epsilon\), \(M_\epsilon=\Omega_{d,\eta_0,\eta_1,\epsilon}(t^2)\) as \(t\to\infty\).
\end{corollary}

{\em Proof.} Recall the refined direct-sum states introduced in the Chernoff converse,
\[
\widetilde\rho_0:=\bigoplus_{i=1}^N\frac1N\rho_i^{(0)},\qquad \widetilde\rho_{1,\pi}:=\bigoplus_{i=1}^N\frac1N\rho_{\pi(i)}^{(1)},\qquad \pi\in S_N,
\]
and let \(P_{e,\pi,\mathrm{ref}}^{(M)}\) denote the optimal Bayes error for testing \(\widetilde\rho_0^{\otimes M}\) against \(\widetilde\rho_{1,\pi}^{\otimes M}\). Coarse-graining the direct-sum index sequence \(\mathbf i=(i_1,\ldots,i_M)\) to its equality pattern \(\operatorname{eq}(\mathbf i)\) recovers the original label-invariant experiment; under \(H_1\), this follows from \(\operatorname{eq}(\pi(\mathbf i))=\operatorname{eq}(\mathbf i)\). Data processing therefore gives
\begin{equation}
P_{e,\mathrm{opt}}^{(M)}\ge P_{e,\pi,\mathrm{ref}}^{(M)},\qquad \pi\in S_N.
\label{eq:original-refined-error-order}
\end{equation}

Recall \(F_{ij}\), \(F_\pi\), and \(F_\star\) from Eq.~\eqref{eq:design-optimal-matching-fidelity}. The tensor powers of the refined states decompose as
\[
\widetilde\rho_0^{\otimes M}=\bigoplus_{\mathbf i\in[N]^M}\frac1{N^M}\bigotimes_{m=1}^M\rho_{i_m}^{(0)},\qquad \widetilde\rho_{1,\pi}^{\otimes M}=\bigoplus_{\mathbf i\in[N]^M}\frac1{N^M}\bigotimes_{m=1}^M\rho_{\pi(i_m)}^{(1)}.
\]
Since the trace norm is additive over direct sums, the Helstrom formula and Lemma~\ref{lem:trace-norm-two-pure} give
\begin{equation}
P_{e,\pi,\mathrm{ref}}^{(M)}=\frac1{2N^M}\sum_{\mathbf i\in[N]^M}\left[1-\sqrt{1-4\eta_0\eta_1\prod_{m=1}^MF_{i_m,\pi(i_m)}}\right].
\label{eq:refined-finite-sample-error}
\end{equation}
Let \(g_\eta(z):=1-\sqrt{1-4\eta_0\eta_1z}\) for \(z\in[0,1]\). This function is increasing and convex because \(g_\eta'(z)=2\eta_0\eta_1(1-4\eta_0\eta_1z)^{-1/2}\ge0\) and \(g_\eta''(z)=4\eta_0^2\eta_1^2(1-4\eta_0\eta_1z)^{-3/2}\ge0\). Using Eq.~\eqref{eq:original-refined-error-order}, Jensen's inequality, and the product structure,
\begin{align}
P_{e,\mathrm{opt}}^{(M)}&\ge\frac12g_\eta\!\left(\frac1{N^M}\sum_{\mathbf i\in[N]^M}\prod_{m=1}^MF_{i_m,\pi(i_m)}\right)\notag\\
&=\frac12g_\eta\!\left[\left(\frac1N\sum_{i=1}^NF_{i,\pi(i)}\right)^M\right]=\frac12\left[1-\sqrt{1-4\eta_0\eta_1F_\pi^M}\right].
\label{eq:refined-Jensen-bound}
\end{align}
Choosing a permutation attaining \(F_\star\) and using \(F_\star\ge x_{\max}\), established in the proof of Theorem~\ref{thm: unversal chernoff t-design upper bound}, proves Eq.~\eqref{eq:finite-sample-design-converse}.

For \(M\le t\), the two uniform \(t\)-designs satisfy, for every \(1\le r\le M\),
\[
\mathcal K_0^{(r)}=N^{1-r}M_0^{(r)}=N^{1-r}\frac{\Pi_{\mathrm{sym}}^{(r)}}{D_r}=N^{1-r}M_1^{(r)}=\mathcal K_1^{(r)}.
\]
The M\"obius expansion in Theorem~\ref{thm:mobius} therefore gives \(\Omega_{0,\Pi}=\Omega_{1,\Pi}=:\Omega_\Pi\) for every \(\Pi\in\Part([M])\). Applying Theorem~\ref{thm:opt} and using \(\sum_{\Pi\in\Part([M])}\Tr\Omega_\Pi=1\), we obtain
\begin{equation}
P_{e,\mathrm{opt}}^{(M)}=\frac12\left[1-|\eta_0-\eta_1|\sum_{\Pi\in\Part([M])}\Tr\Omega_\Pi\right]=\frac12(1-|\eta_0-\eta_1|)=\min\{\eta_0,\eta_1\}.
\label{eq:design-exact-indistinguishability}
\end{equation}

Finally, suppose \(0<\epsilon<\min\{\eta_0,\eta_1\}\). If \(M_\epsilon=+\infty\), Eq.~\eqref{eq:finite-sample-design-complexity-bound} is immediate. Otherwise, Eq.~\eqref{eq:design-exact-indistinguishability} implies \(M_\epsilon\ge t+1\), while Eq.~\eqref{eq:finite-sample-design-converse} gives
\[
\epsilon\ge\frac12\left[1-\sqrt{1-4\eta_0\eta_1x_{\max}^{M_\epsilon}}\right].
\]
Equivalently, \(x_{\max}^{M_\epsilon}\le\epsilon(1-\epsilon)/(\eta_0\eta_1)<1\), and taking logarithms yields Eq.~\eqref{eq:finite-sample-design-complexity-bound}. Finally, Eq.~\eqref{eq:Jacobi-design-root-asymptotics} gives \(-\log x_{\max}=O_d(t^{-2})\). Since \(\log[\eta_0\eta_1/(\epsilon(1-\epsilon))]>0\) is independent of \(t\), it follows that \(M_\epsilon=\Omega_{d,\eta_0,\eta_1,\epsilon}(t^2)\). \hfill\(\square\)

\paragraph{\textbf{Proof of order-achievability of the \(t^{-2}\) scaling.}} Let \(\calE=\{(1/N,\rho_i)\}_{i=1}^N\) be a finite uniform pure-state \(t\)-design on \(\mathbb C^d\). For \(U\in U(d)\), define the optimal matching fidelity and infidelity between \(\calE\) and \(U\calE:=\{(1/N,U\rho_iU^\dagger)\}_{i=1}^N\) by
\begin{equation}
F_\star(U):=\max_{\pi\in S_N}\frac1N\sum_{i=1}^N\Tr\!\left(\rho_iU\rho_{\pi(i)}U^\dagger\right),\qquad \varepsilon_\star(U):=1-F_\star(U)=\min_{\pi\in S_N}\frac1N\sum_{i=1}^N\left[1-\Tr\!\left(\rho_iU\rho_{\pi(i)}U^\dagger\right)\right].
\label{eq:rotated-design-matching}
\end{equation}
For every \(U\in U(d)\), the rotated ensemble \(U\calE\) remains a uniform pure-state \(t\)-design, since for \(1\le r\le t\),
\[
\frac1N\sum_{i=1}^N(U\rho_iU^\dagger)^{\otimes r}=U^{\otimes r}\left(\frac1N\sum_{i=1}^N\rho_i^{\otimes r}\right)(U^\dagger)^{\otimes r}=U^{\otimes r}\frac{\Pi_{\mathrm{sym}}^{(r)}}{D_r}(U^\dagger)^{\otimes r}=\frac{\Pi_{\mathrm{sym}}^{(r)}}{D_r},
\]
where \(D_r=\binom{d+r-1}{r}\) and \(\Pi_{\mathrm{sym}}^{(r)}\) is invariant under \(U^{\otimes r}\). Since \(\calE\) and \(U\calE\) are finite uniform pure-state ensembles, Theorem~\ref{thm:fixed-N-chernoff-bounds} gives
\begin{equation}
E(\calE,U\calE)=-\log F_\star(U)=-\log[1-\varepsilon_\star(U)].
\label{eq:rotated-design-exponent}
\end{equation}

For every \(\pi\in S_N\), each matched cost is at least the minimum cost in the corresponding row. Hence,
\begin{equation}
\varepsilon_\star(U)\ge\frac1N\sum_{i=1}^N\min_{1\le j\le N}\left[1-\Tr\!\left(\rho_iU\rho_jU^\dagger\right)\right].
\label{eq:matching-nearest-point-bound}
\end{equation}
Let \(U\) be Haar distributed on \(U(d)\). By cyclicity of the trace and linearity of expectation,
\[
\mathbb E_U\varepsilon_\star(U)\ge\frac1N\sum_{i=1}^N\mathbb E_U\min_{1\le j\le N}\left[1-\Tr\!\left(U^\dagger\rho_iU\rho_j\right)\right].
\]
For each fixed \(i\), \(U^\dagger\rho_iU\) is Haar distributed over pure states. Therefore, defining \(\rho_\psi:=\ket{\psi}\!\bra{\psi}\) and
\[
Z(\psi):=\min_{1\le j\le N}\left[1-\Tr(\rho_\psi\rho_j)\right],
\]
we obtain
\begin{equation}
\mathbb E_U\varepsilon_\star(U)\ge\int d\psi\,Z(\psi).
\label{eq:rotation-average-nearest-infidelity}
\end{equation}

For every fixed pure state \(\rho_j\) and \(0\le s\le1\),
\begin{equation}
\Pr_\psi\!\left[1-\Tr(\rho_\psi\rho_j)\le s\right]=s^{d-1}.
\label{eq:pure-state-cap-volume}
\end{equation}
To verify Eq.~\eqref{eq:pure-state-cap-volume}, unitary invariance allows us to take \(\rho_j=\ket{1}\!\bra{1}\). Let \(g_1,\ldots,g_d\) be independent standard complex Gaussian variables and set \(\ket{\psi}=(\sum_{k=1}^d|g_k|^2)^{-1/2}\sum_{k=1}^dg_k\ket{k}\). Then \(\ket{\psi}\) is Haar distributed. With \(A:=|g_1|^2\) and \(B:=\sum_{k=2}^d|g_k|^2\), the variables \(A\sim\operatorname{Gamma}(1,1)\) and \(B\sim\operatorname{Gamma}(d-1,1)\) are independent, and \(1-\Tr(\rho_\psi\rho_j)=B/(A+B)\). Thus, for \(0<s\le1\),
\[
\Pr_\psi\!\left[1-\Tr(\rho_\psi\rho_j)\le s\right]=\Pr\!\left[A\ge\frac{1-s}{s}B\right]=\mathbb E_B\exp\!\left[-\frac{1-s}{s}B\right]=\left(1+\frac{1-s}{s}\right)^{-(d-1)}=s^{d-1},
\]
while the case \(s=0\) follows by continuity.

The union bound and Eq.~\eqref{eq:pure-state-cap-volume} give \(\Pr_\psi[Z(\psi)\le s]\le Ns^{d-1}\). Since \(0\le Z(\psi)\le1\), the tail-integral identity yields
\begin{align}
\int d\psi\,Z(\psi)&=\int_0^1\Pr_\psi[Z(\psi)>s]\,ds\notag\\
&\ge\int_0^{N^{-1/(d-1)}}(1-Ns^{d-1})\,ds=\frac{d-1}{d}N^{-1/(d-1)}.
\label{eq:average-nearest-infidelity-bound}
\end{align}
Combining Eqs.~\eqref{eq:rotation-average-nearest-infidelity} and \eqref{eq:average-nearest-infidelity-bound}, there exists \(U\in U(d)\) such that \(\varepsilon_\star(U)\ge\frac{d-1}{d}N^{-1/(d-1)}\). Equation~\eqref{eq:rotated-design-exponent} then gives the finite-cardinality bound
\begin{equation}
E(\calE,U\calE)\ge-\log\!\left[1-\frac{d-1}{d}N^{-1/(d-1)}\right]\ge\frac{d-1}{d}N^{-1/(d-1)},
\label{eq:finite-cardinality-design-achievability}
\end{equation}
where the second inequality follows from \(-\log(1-x)\ge x\) for \(0\le x<1\).

It remains to choose a \(t\)-design with sufficiently small cardinality. Consider the rank-one projector manifold
\[
\mathcal P_d:=\{\rho\in\operatorname{Herm}(\mathbb C^d):\rho^2=\rho,\ \Tr\rho=1\}.
\]
The manifold \(\mathcal P_d\) is smooth, connected, compact, real algebraic, isomorphic to \(\mathbb{CP}^{d-1}\), and has real dimension \(2(d-1)\). Endow it with the Riemannian volume induced by the Hilbert--Schmidt metric. Since the conjugation action of \(U(d)\) is transitive and isometric, the normalized volume measure coincides with the unitarily invariant probability measure on pure states.

By Theorem~2.2 of Ref.~\cite{Etayo2018asymptoticallyoptimaldesigns}, there exists \(C_{\mathcal P_d}>0\) such that, for every integer \(N\ge C_{\mathcal P_d}t^{2(d-1)}\), the manifold \(\mathcal P_d\) admits an equal-weight polynomial \(t\)-design with \(N\) points. Set
\begin{equation}
N_t:=\left\lceil C_{\mathcal P_d}t^{2(d-1)}\right\rceil\le K_dt^{2(d-1)},\qquad K_d:=C_{\mathcal P_d}+1.
\label{eq:state-design-cardinality}
\end{equation}
Let \(\{\rho_i^{(t)}\}_{i=1}^{N_t}\subset\mathcal P_d\) be such a polynomial \(t\)-design. Each matrix element of \(\rho^{\otimes r}\) is a polynomial of degree \(r\) in the real coordinates of \(\rho\); applying the design identity to its real and imaginary parts gives, for \(1\le r\le t\),
\begin{equation}
\frac1{N_t}\sum_{i=1}^{N_t}\left(\rho_i^{(t)}\right)^{\otimes r}=\int d\psi\,\rho_\psi^{\otimes r}=\frac{\Pi_{\mathrm{sym}}^{(r)}}{D_r}.
\label{eq:polynomial-design-to-state-design}
\end{equation}
Hence \(\calE_t:=\{(1/N_t,\rho_i^{(t)})\}_{i=1}^{N_t}\) is a finite uniform pure-state \(t\)-design.

Applying Eq.~\eqref{eq:finite-cardinality-design-achievability} to \(\calE_t\), there exists \(U_t\in U(d)\) such that
\[
E(\calE_t,U_t\calE_t)\ge-\log\!\left[1-\frac{d-1}{d}N_t^{-1/(d-1)}\right].
\]
Since Eq.~\eqref{eq:state-design-cardinality} implies \(N_t^{-1/(d-1)}\ge K_d^{-1/(d-1)}t^{-2}\), defining \(c_d:=\frac{d-1}{d}K_d^{-1/(d-1)}>0\) gives
\begin{equation}
E(\calE_t,U_t\calE_t)\ge-\log(1-c_dt^{-2})\ge c_dt^{-2}.
\label{eq:t-design-exponent-achievability}
\end{equation}
This proves order-achievability of the \(t^{-2}\) scaling. \hfill\(\square\)

\subsection{Explicit examples of testing between t-designs}

Here we illustrate the explicit examples of testing between $t$-designs shown in Fig.~1(b). For a unit Bloch vector $\bm r\in S^2$, write
\begin{equation}
\rho(\bm r)=\frac{I+\bm r\cdot\bm\sigma}{2}.
\end{equation}
An equal-weight point set $D_t=\{\bm r_i\}_{i=1}^{N_t}$ defines a qubit $t$-design if
\begin{equation}
\frac{1}{N_t}\sum_{i=1}^{N_t}\rho(\bm r_i)^{\otimes k}=\frac{\Pi_{\mathrm{sym}}^{(k)}}{k+1},\qquad 1\leq k\leq t.
\label{eq:app-design-condition}
\end{equation}
For each example, the two hypotheses are
\begin{align}
\mathcal E_{t,0}&=\left\{\left(\frac1{N_t},\rho(\bm r_i)\right)\right\}_{i=1}^{N_t},\notag\\
\mathcal E_{t,1}&=\left\{\left(\frac1{N_t},\rho(R_t\bm r_i)\right)\right\}_{i=1}^{N_t},
\label{eq:app-explicit-pair}
\end{align}
where $R_t\in SO(3)$. Rotations preserve Eq.~\eqref{eq:app-design-condition}, so both hypotheses are $t$-designs.

{\em \textbf{Calculation of the Chernoff exponent.}}
For a relabeling $\pi\in S_{N_t}$, define
\begin{equation}
\xi_{t,\pi}(s)=-\log\left[\frac1{N_t}\sum_{i=1}^{N_t}\operatorname{Tr}\left(\rho(\bm r_i)^s\rho(R_t\bm r_{\pi(i)})^{1-s}\right)\right].
\label{eq:app-fixed-permutation-xi}
\end{equation}
All component states are rank-one projectors. Hence $\rho^s=\rho$ for every $0\leq s\leq1$, with the standard support convention at the endpoints, and Eq.~\eqref{eq:app-fixed-permutation-xi} is independent of $s$. Define the cross-fidelity matrix
\begin{equation}
F_{ij}^{(t)}:=\operatorname{Tr}\!\left[\rho(\bm r_i)\rho(R_t\bm r_j)\right]=\frac{1+\bm r_i^{\mathsf T}R_t\bm r_j}{2}.
\label{eq:app-cross-fidelity}
\end{equation}
For finite uniform pure-state ensembles of equal cardinality, the Chernoff exponent is therefore
\begin{align}
E_t&=\min_{\pi\in S_{N_t}}\sup_{0\leq s\leq1}\xi_{t,\pi}(s)=-\log F_t^\star,
\label{eq:app-explicit-exponent}\\
F_t^\star&:=\max_{\pi\in S_{N_t}}\frac1{N_t}\sum_{i=1}^{N_t}F_{i,\pi(i)}^{(t)}.
\label{eq:app-Fstar}
\end{align}
Thus, $N_tF_t^\star$ is the maximum-weight perfect-matching value of $F^{(t)}$. For each pair below, we construct $F^{(t)}$ using Eq.~\eqref{eq:app-cross-fidelity}, solve Eq.~\eqref{eq:app-Fstar} with the Hungarian algorithm, and evaluate $E_t=-\log F_t^\star$.

\paragraph{$\mathbf{t=1}$.}
Let
\begin{equation}
D_1=\{\bm e_z,-\bm e_z\},\qquad R_1=R_y(\pi/2)=\begin{pmatrix}0&0&1\\0&1&0\\-1&0&0\end{pmatrix}.
\end{equation}
Every cross fidelity equals $1/2$. Therefore,
\begin{equation}
F_1^\star=\frac12,\qquad E_1=\log2=0.693147180559945.
\end{equation}

\paragraph{$\mathbf {t=2}$.}
Use the tetrahedral point set
\begin{equation}
D_2=\frac1{\sqrt3}\bigl\{(1,1,1),(1,-1,-1),(-1,1,-1),(-1,-1,1)\bigr\}
\end{equation}
and
\begin{equation}
R_2=R_z(\pi/2)=\begin{pmatrix}0&-1&0\\1&0&0\\0&0&1\end{pmatrix}.
\end{equation}
The cross fidelities are either $0$ or $2/3$, and there is a perfect matching containing only entries equal to $2/3$. Hence,
\begin{equation}
F_2^\star=\frac23,\qquad E_2=\log\frac32=0.405465108108164.
\end{equation}

\paragraph{$\mathbf{t=3}$.}
Use the octahedral point set
\begin{equation}
D_3=\{\bm e_x,\bm e_y,\bm e_z,-\bm e_x,-\bm e_y,-\bm e_z\}
\end{equation}
and
\begin{equation}
R_3=\begin{pmatrix}1/\sqrt2&0&1/\sqrt2\\-1/2&1/\sqrt2&1/2\\-1/2&-1/\sqrt2&1/2\end{pmatrix}.
\end{equation}
In the displayed ordering, an optimal assignment contains four fidelities $A=(2+\sqrt2)/4$ and two fidelities equal to $3/4$. Therefore,
\begin{equation}
F_3^\star=\frac{4A+2(3/4)}6=\frac{7+2\sqrt2}{12}=0.819035593728849,
\end{equation}
and
\begin{equation}
E_3=-\log\frac{7+2\sqrt2}{12}=0.199627736087041.
\end{equation}

\paragraph{$\mathbf{t=4}$.}
Define
\begin{equation}
a=\sqrt{\frac{1+2/\sqrt5}{3}},\qquad b=\sqrt{\frac{1-2/\sqrt5}{3}},\qquad \phi_k=\frac{2\pi k}{5},
\end{equation}
and take the $20$ points
\begin{equation}
\begin{aligned}
D_4=\bigl\{&\bigl(\sqrt{1-z^2}\cos\phi_k,\sqrt{1-z^2}\sin\phi_k,z\bigr):\\
&z\in\{a,-a,b,-b\},\quad 0\leq k\leq4\bigr\}.
\end{aligned}
\label{eq:app-D4}
\end{equation}
The four latitude nodes reproduce the uniform $z$-moments through degree four, while each regular pentagon eliminates all nonzero azimuthal Fourier modes of order at most four. Thus, $D_4$ is an exact equal-weight spherical $4$-design. The numerical rotation used to generate the $t=4$ cross in Fig.~1(b) is
{\scriptsize
\begin{equation}
R_4=\begin{pmatrix}
0.801193425832&-0.563706784185&0.200807758487\\
-0.593469407560&-0.791522963285&0.145895376494\\
0.076701738537&-0.236063677969&-0.968705725826
\end{pmatrix}.
\end{equation}
}
The resulting $20\times20$ assignment gives
\begin{equation}
F_4^\star=0.931768454373241,\qquad E_4=-\log F_4^\star=0.070670934676444.
\end{equation}

\paragraph{$\mathbf{t=5}$.}
Let $\varphi=(1+\sqrt5)/2$ and take the $12$ normalized icosahedral vertices
\begin{equation}
\begin{aligned}
D_5=\frac1{\sqrt{1+\varphi^2}}\bigl\{&(0,s_1,s_2\varphi),(s_1,s_2\varphi,0),\\
&(s_2\varphi,0,s_1):s_1,s_2\in\{-1,1\}\bigr\}.
\end{aligned}
\label{eq:app-D5}
\end{equation}
This is an exact equal-weight spherical $5$-design. The numerical rotation used to generate the $t=5$ cross in Fig.~1(b) is
{\scriptsize
\begin{equation}
R_5=\begin{pmatrix}
-0.086426947665&-0.704460902630&-0.704460942411\\
0.496265199926&0.582692136870&-0.643576510581\\
0.863858341431&-0.405221803824&0.299239127866
\end{pmatrix}.
\end{equation}
}
The corresponding $12\times12$ assignment gives
\begin{equation}
F_5^\star=0.931065370145479,\qquad E_5=-\log F_5^\star=0.071425789190994.
\end{equation}

For comparison, the universal qubit bound plotted in Fig.~1(b) is
\begin{equation}
E_t\leq\overline E_{2,t}=-\log x_t,\qquad x_t=\frac{1+z_t}{2},
\end{equation}
where $z_t$ is the largest zero of $P_{\lceil t/2\rceil}^{(0,\beta_t)}(z)$, with $\beta_t=0$ for odd $t$ and $\beta_t=1$ for even $t$. The following values are the values used to generate the five crosses and the corresponding bound points in Fig.~1(b):
\begin{center}
\small
\setlength{\tabcolsep}{3pt}
\begin{tabular}{c c c c c}
\hline
$t$ & $N_t$ & $F_t^\star$ & $E_t$ & $\overline E_{2,t}$\\
\hline
1 & 2  & 0.500000000 & 0.693147181 & 0.693147181\\
2 & 4  & 0.666666667 & 0.405465108 & 0.405465108\\
3 & 6  & 0.819035594 & 0.199627736 & 0.237400786\\
4 & 20 & 0.931768454 & 0.070670935 & 0.168479039\\
5 & 12 & 0.931065370 & 0.071425789 & 0.119574012\\
\hline
\end{tabular}
\end{center}

The $t=1$ and $t=2$ pairs saturate the universal bound. The red reference curve in Fig.~1(b) is $c_{\mathrm{fit}}/t^2$, where $c_{\mathrm{fit}}=5.6783083115$ is obtained from the bound values over $10\leq t\leq1000$ and is displayed in the legend as $5.678/t^2$. For each fixed $R_t$, Eq.~\eqref{eq:app-Fstar} is solved globally. The rotations for $t=4,5$ are fixed numerical choices obtained from a search over $SO(3)$. We do not claim global optimality over all relative rotations for $t=3,4,5$; the corresponding crosses in Fig.~1(b) are the achieved Chernoff exponents of the displayed pairs.

\subsection{Symmetric two-component ensembles and optical constellation testing}

\paragraph{\textbf {Symmetric two-component ensembles.}}

Consider two uniform pure-state ensembles
\begin{equation}
\mathcal E_h=\left\{\left(\frac12,\ketbra{\psi^{(h)}_1}\right),\left(\frac12,\ketbra{\psi^{(h)}_2}\right)\right\},\qquad h=0,1,
\end{equation}
with hypothesis priors $\eta_0,\eta_1>0$ and $\eta_0+\eta_1=1$. Assume
\begin{equation}
\begin{aligned}
\braket{\psi^{(0)}_1}{\psi^{(0)}_2}&=\braket{\psi^{(1)}_1}{\psi^{(1)}_2}=x\in\mathbb R,\qquad |x|<1,\\
\braket{\psi^{(0)}_1}{\psi^{(1)}_1}&=\braket{\psi^{(0)}_2}{\psi^{(1)}_2}=y,\qquad \braket{\psi^{(0)}_1}{\psi^{(1)}_2}=\braket{\psi^{(0)}_2}{\psi^{(1)}_1}=y^\star.
\end{aligned}
\label{eq:app-symmetric-overlaps}
\end{equation}

Every possible branch $\Pi\in\Part([M])$ has one or two blocks. For a two-block branch, choose an ordering $\Pi=\{B_1,B_2\}$; for the one-block branch, set $B_1=[M]$ and $B_2=\varnothing$ for notational uniformity. Let $a:=|B_1|$, so that $|B_2|=M-a$. After conjugating by the register-permutation unitary that groups the registers in $B_1$ and $B_2$, define
\begin{equation}
\ket{\Phi^{(h)}_\Pi}:=\ket{\psi^{(h)}_1}^{\otimes a}\otimes\ket{\psi^{(h)}_2}^{\otimes(M-a)},\qquad \ket{\Psi^{(h)}_\Pi}:=\ket{\psi^{(h)}_2}^{\otimes a}\otimes\ket{\psi^{(h)}_1}^{\otimes(M-a)},
\end{equation}
where an empty tensor product is understood as the scalar $1$. The two compatible assignments of the component labels to the observed blocks give
\begin{equation}
\Omega_{h,\Pi}=2^{-M}\left(\ketbra{\Phi^{(h)}_\Pi}+\ketbra{\Psi^{(h)}_\Pi}\right),
\label{eq:app-symmetric-branch-operator}
\end{equation}
and hence $\Tr\Omega_{h,\Pi}=2^{1-M}$ for every possible branch.

The symmetry assumptions imply
\begin{equation}
\braket{\Phi^{(h)}_\Pi}{\Psi^{(h)}_\Pi}=x^a(x^\star)^{M-a}=x^M,
\end{equation}
as well as
\begin{equation}
\braket{\Phi^{(0)}_\Pi}{\Phi^{(1)}_\Pi}=\braket{\Psi^{(0)}_\Pi}{\Psi^{(1)}_\Pi}=y^M,\qquad \braket{\Phi^{(0)}_\Pi}{\Psi^{(1)}_\Pi}=\braket{\Psi^{(0)}_\Pi}{\Phi^{(1)}_\Pi}=(y^\star)^M.
\label{eq:app-symmetric-branch-overlaps}
\end{equation}
In particular, all relevant inner products are independent of the block size $a$ and hence of the branch $\Pi$. Introduce the normalized symmetric and antisymmetric vectors
\begin{equation}
\ket{\phi^{(h)}_{\Pi,\pm}}:=\frac{\ket{\Phi^{(h)}_\Pi}\pm\ket{\Psi^{(h)}_\Pi}}{\sqrt{2(1\pm x^M)}}.
\end{equation}
They obey $\braket{\phi^{(h)}_{\Pi,+}}{\phi^{(h)}_{\Pi,-}}=0$, and Eq.~\eqref{eq:app-symmetric-branch-operator} becomes
\begin{equation}
\Omega_{h,\Pi}=2^{-M}\left[(1+x^M)\ketbra{\phi^{(h)}_{\Pi,+}}+(1-x^M)\ketbra{\phi^{(h)}_{\Pi,-}}\right].
\label{eq:app-symmetric-spectral-decomposition}
\end{equation}
Equation~\eqref{eq:app-symmetric-branch-overlaps} further gives
\begin{equation}
\braket{\phi^{(0)}_{\Pi,+}}{\phi^{(0)}_{\Pi,-}}=\braket{\phi^{(1)}_{\Pi,-}}{\phi^{(1)}_{\Pi,+}}=
\braket{\phi^{(0)}_{\Pi,+}}{\phi^{(1)}_{\Pi,-}}=\braket{\phi^{(0)}_{\Pi,-}}{\phi^{(1)}_{\Pi,+}}=0.
\label{eq:app-symmetric-parity-overlaps}
\end{equation}

Therefore,
\begin{equation}
\Delta_\Pi:=\eta_0\Omega_{0,\Pi}-\eta_1\Omega_{1,\Pi}
\end{equation}
is decomposed into two orthogonal sectors:
\begin{equation}
\begin{aligned}
\Delta_\Pi=2^{-M}\Big[&(1+x^M)\left(\eta_0\ketbra{\phi^{(0)}_{\Pi,+}}-\eta_1\ketbra{\phi^{(1)}_{\Pi,+}}\right)\\
&+(1-x^M)\left(\eta_0\ketbra{\phi^{(0)}_{\Pi,-}}-\eta_1\ketbra{\phi^{(1)}_{\Pi,-}}\right)\Big].
\end{aligned}
\end{equation}
Then the trace norm can be calculated as the sum of two terms:
\begin{equation}
\begin{aligned}
\norm{\Delta_\Pi}_1=2^{-M}\Big[&(1+x^M)\norm{\eta_0\ketbra{\phi^{(0)}_{\Pi,+}}-\eta_1\ketbra{\phi^{(1)}_{\Pi,+}}}_1\\
&+(1-x^M)\norm{\eta_0\ketbra{\phi^{(0)}_{\Pi,-}}-\eta_1\ketbra{\phi^{(1)}_{\Pi,-}}}_1\Big].
\end{aligned}
\end{equation}

For normalized vectors $\ket{u}$ and $\ket{v}$, by Lemma~\ref{lem:trace-norm-two-pure} we have
\begin{equation}
\left\|\eta_0\ketbra{u}-\eta_1\ketbra{v}\right\|_1=\sqrt{1-4\eta_0\eta_1|\braket{u}{v}|^2}.
\label{eq:app-weighted-pure-trace-norm}
\end{equation}
Applying Eq.~\eqref{eq:app-weighted-pure-trace-norm} separately to the two orthogonal sectors yields (using $g_\pm :=\braket{\phi^{(0)}_{\Pi,\pm}}{\phi^{(1)}_{\Pi,\pm}}=\frac{y^M\pm(y^\star)^M}{1\pm x^M}$)
\begin{equation}
\begin{aligned}
\norm{\Delta_\Pi}_1=2^{-M}\Bigg[&(1+x^M)\sqrt{1-4\eta_0\eta_1\left|\frac{y^M+(y^\star)^M}{1+x^M}\right|^2}\\
&+(1-x^M)\sqrt{1-4\eta_0\eta_1\left|\frac{y^M-(y^\star)^M}{1-x^M}\right|^2}\Bigg].
\end{aligned}
\label{eq:app-symmetric-branch-trace-norm}
\end{equation}

There is one one-block branch and $(2^M-2)/2=2^{M-1}-1$ two-block branches, and hence $2^{M-1}$ possible branches in total. Since Eq.~\eqref{eq:app-symmetric-branch-trace-norm} is independent of $\Pi$, Eq.~\eqref{eq:Bayes-exact} gives, with $\lambda_\pm:=(1\pm x^M)/2$,
\begin{equation}
\begin{aligned}
P_{e,\mathrm{opt}}^{(M)}=\frac12\Bigg[1&-\lambda_+\sqrt{1-4\eta_0\eta_1\left|\frac{y^M+(y^\star)^M}{1+x^M}\right|^2}\\
&-\lambda_-\sqrt{1-4\eta_0\eta_1\left|\frac{y^M-(y^\star)^M}{1-x^M}\right|^2}\Bigg].
\end{aligned}
\label{eq:app-symmetric-two-error}
\end{equation}
This reproduces Eq.~\eqref{PE_symmetric_two}. For $0<|y|<1$, expanding $1-\sqrt{1-t}=t/2+O(t^2)$ and using Eq.~\eqref{eq:app-symmetric-parity-overlaps} gives
\begin{equation}
P_{e,\mathrm{opt}}^{(M)}=\eta_0\eta_1\left(\lambda_+|g_+|^2+\lambda_-|g_-|^2\right)+O(|y|^{4M}),
\end{equation}
where
\begin{equation}
\lambda_+|g_+|^2+\lambda_-|g_-|^2=\frac12\left[\frac{|y^M+(y^\star)^M|^2}{1+x^M}+\frac{|y^M-(y^\star)^M|^2}{1-x^M}\right]=2|y|^{2M}\left[1+O(|x|^M)\right].
\end{equation}
Therefore
\begin{equation}
P_{e,\mathrm{opt}}^{(M)}=2\eta_0\eta_1|y|^{2M}\left[1+O(|x|^M+|y|^{2M})\right],\qquad E=-2\log|y|.
\end{equation}
If $y=0$, Eq.~\eqref{eq:app-symmetric-two-error} gives $P_{e,\mathrm{opt}}^{(M)}=0$ for every $M\ge1$.

\paragraph{\textbf{Optical communication systems.}}

In a binary phase-shift keying system, exchanging the bit labels $0$ and $1$ does not change the physical signal constellation. The real-axis and imaginary-axis constellations are therefore represented by
\begin{equation}
\mathcal E_R=\left\{\left(\frac12,\ketbra{\alpha}\right),\left(\frac12,\ketbra{-\alpha}\right)\right\},\qquad \mathcal E_I=\left\{\left(\frac12,\ketbra{-i\alpha}\right),\left(\frac12,\ketbra{i\alpha}\right)\right\}.
\label{eq:app-optical-ensembles}
\end{equation}
For coherent states, $\braket{\beta}{\gamma}=\exp[-(|\beta|^2+|\gamma|^2)/2+\beta^\star\gamma]$. Writing $\bar n:=|\alpha|^2$, the ordering in Eq.~\eqref{eq:app-optical-ensembles} gives
\begin{equation}
x=\braket{\alpha}{-\alpha}=\braket{-i\alpha}{i\alpha}=e^{-2\bar n},\qquad y=\braket{\alpha}{-i\alpha}=\braket{-\alpha}{i\alpha}=e^{-(1+i)\bar n}.
\label{eq:app-optical-xy}
\end{equation}
Thus the squared overlap within either constellation is $|\braket{\alpha}{-\alpha}|^2=|\braket{-i\alpha}{i\alpha}|^2=e^{-4\bar n}$, whereas every cross-constellation squared overlap satisfies $|\braket{\pm\alpha}{\pm i\alpha}|^2=|\braket{\pm\alpha}{\mp i\alpha}|^2=e^{-2\bar n}$. Setting $z:=M\bar n$, we have $x^M=e^{-2z}$ and $y^M=e^{-(1+i)z}$, so that
\begin{equation}
\left|\frac{y^M+(y^\star)^M}{1+x^M}\right|^2=\frac{\cos^2z}{\cosh^2z},\qquad \left|\frac{y^M-(y^\star)^M}{1-x^M}\right|^2=\frac{\sin^2z}{\sinh^2z}.
\label{eq:app-optical-ratios}
\end{equation}
Indeed, $y^M+(y^\star)^M=2e^{-z}\cos z$, $y^M-(y^\star)^M=-2ie^{-z}\sin z$, $1+e^{-2z}=2e^{-z}\cosh z$, and $1-e^{-2z}=2e^{-z}\sinh z$. Substituting Eq.~\eqref{eq:app-optical-ratios} into Eq.~\eqref{eq:app-symmetric-two-error} gives
\begin{equation}
\begin{aligned}
P_{e,\mathrm{opt}}^{(M)}=\frac12\Bigg[1&-\frac{1+e^{-2z}}{2}\sqrt{1-4\eta_0\eta_1\frac{\cos^2z}{\cosh^2z}}\\
&-\frac{1-e^{-2z}}{2}\sqrt{1-4\eta_0\eta_1\frac{\sin^2z}{\sinh^2z}}\Bigg],\qquad z=M|\alpha|^2.
\end{aligned}
\label{eq:app-coherent-BPSK-error}
\end{equation}
To obtain the large-$z$ behavior, set $q:=e^{-2z}$, $R_+:=\cos^2z/\cosh^2z$, and $R_-:=\sin^2z/\sinh^2z$. Then $R_\pm=O(q)$ and
\begin{equation}
\frac{1+q}{2}R_++\frac{1-q}{2}R_-=\frac{2q[1-q\cos(2z)]}{1-q^2}=2q+O(q^2).
\end{equation}
Using $1-\sqrt{1-t}=t/2+O(t^2)$ in Eq.~\eqref{eq:app-coherent-BPSK-error} therefore gives
\begin{equation}
P_{e,\mathrm{opt}}^{(M)}=2\eta_0\eta_1e^{-2M|\alpha|^2}+O\!\left(e^{-4M|\alpha|^2}\right),\qquad E=2|\alpha|^2.
\end{equation}

\subsection{Two common pure component states}

We now specialize to two ensembles with common pure components:
\[
H_0:\calE_p=\{(p,\rho_1),(1-p,\rho_2)\},\qquad H_1:\calE_q=\{(q,\rho_1),(1-q,\rho_2)\},
\]
where \(0<p,q<1\), \(\rho_1=\ket{\psi_1}\!\bra{\psi_1}\), \(\rho_2=\ket{\psi_2}\!\bra{\psi_2}\), and \(c:=|\langle\psi_1|\psi_2\rangle|^2\in[0,1]\). Throughout this subsection, the hypothesis priors are equal.

Consider a two-block branch \(\Pi=\{B,B^c\}\) with \(|B|=r\) and \(|B^c|=M-r\). Define
\[
\tau_{12}^{B}:=\bigotimes_{t\in B}\rho_1^{(t)}\otimes\bigotimes_{t\in B^c}\rho_2^{(t)},\qquad \tau_{21}^{B}:=\bigotimes_{t\in B}\rho_2^{(t)}\otimes\bigotimes_{t\in B^c}\rho_1^{(t)}.
\]
Then
\[
\Omega_{p,\Pi}-\Omega_{q,\Pi}=A_{r,M-r}\tau_{12}^{B}+B_{r,M-r}\tau_{21}^{B},
\]
where
\[
A_{r,M-r}:=p^r(1-p)^{M-r}-q^r(1-q)^{M-r},\qquad B_{r,M-r}:=(1-p)^rp^{M-r}-(1-q)^rq^{M-r}.
\]
The squared overlap between the two product states is independent of \(r\):
\[
\Tr(\tau_{12}^{B}\tau_{21}^{B})=|\langle\psi_1|\psi_2\rangle|^{2r}|\langle\psi_2|\psi_1\rangle|^{2(M-r)}=c^M.
\]
For \(z\in[0,1]\), define
\begin{equation}
\Phi_z(A,B):=\begin{cases}|A+B|,&AB\ge0,\\[1mm]\sqrt{(A-B)^2+4ABz},&AB<0.\end{cases}
\label{eq:two-pure-trace-norm-function}
\end{equation}
Lemma~\ref{lem:trace-norm-two-pure} therefore gives
\[
\|\Omega_{p,\Pi}-\Omega_{q,\Pi}\|_1=\Phi_{c^M}(A_{r,M-r},B_{r,M-r}).
\]

The one-block branch \(\Pi=\{[M]\}\) is included in the same notation by setting \(r=0\), \(\tau_{12}^{0,M}:=\rho_2^{\otimes M}\), \(\tau_{21}^{0,M}:=\rho_1^{\otimes M}\), \(A_{0,M}:=(1-p)^M-(1-q)^M\), and \(B_{0,M}:=p^M-q^M\). Grouping the branches by unordered block sizes, define
\[
\nu_{M,r}:=\begin{cases}1,&r=0,\\[1mm]\binom Mr,&1\le r<M/2,\\[1mm]\frac12\binom M{M/2},&r=M/2\text{ and }M\text{ is even}.\end{cases}
\]
The exact equal-prior Bayes error is then
\begin{equation}
P_{e,\mathrm{opt}}^{(M)}=\frac12\left[1-\frac12\sum_{r=0}^{\lfloor M/2\rfloor}\nu_{M,r}\Phi_{c^M}(A_{r,M-r},B_{r,M-r})\right].
\label{eq:two-state-common-pure-bayes}
\end{equation}
Here \(A_{r,M-r}\) and \(B_{r,M-r}\) are the signed likelihood differences associated with the two assignments of the component states to the blocks, while \(c^M\) is the quantum overlap between the corresponding product-state assignments.

For \(u,v\in(0,1)\), define the Bernoulli Chernoff information by
\[
\xi_{\mathrm B}(u,v):=-\log\min_{0\le s\le1}\left[u^sv^{1-s}+(1-u)^s(1-v)^{1-s}\right].
\]
There are two cross-hypothesis matchings. For the identity matching,
\[
C_{\mathrm{id}}(s)=p^sq^{1-s}+(1-p)^s(1-q)^{1-s},\qquad \xi_{\mathrm{id}}=\xi_{\mathrm B}(p,q).
\]
For the swap matching,
\[
C_{\mathrm{sw}}(s)=c\left[p^s(1-q)^{1-s}+(1-p)^sq^{1-s}\right],
\]
where \(\Tr(\rho_1^s\rho_2^{1-s})=\Tr(\rho_2^s\rho_1^{1-s})=c\). Hence
\[
\xi_{\mathrm{sw}}=-\log c+\xi_{\mathrm B}(p,1-q),
\]
with the convention \(-\log0=+\infty\), and
\begin{equation}
\xi_{\mathrm{acr}}=\min\{\xi_{\mathrm{id}},\xi_{\mathrm{sw}}\}=\min\left\{\xi_{\mathrm B}(p,q),-\log c+\xi_{\mathrm B}(p,1-q)\right\}.
\label{eq:two-common-pure-xi-acr}
\end{equation}
The fixed-permutation Chernoff converse gives \(E\le\xi_{\mathrm{acr}}\). We next prove the matching achievability bound.

For a branch with unordered block size \(r\), set
\[
x_r:=p^r(1-p)^{M-r},\qquad y_r:=(1-p)^rp^{M-r},\qquad u_r:=q^r(1-q)^{M-r},\qquad v_r:=(1-q)^rq^{M-r},
\]
so that \(A_{r,M-r}=x_r-u_r\) and \(B_{r,M-r}=y_r-v_r\). The contribution of one such branch to the equal-prior Bayes error is
\[
e_{M,r}:=\frac14\left(\Tr\Omega_{p,\Pi}+\Tr\Omega_{q,\Pi}-\|\Omega_{p,\Pi}-\Omega_{q,\Pi}\|_1\right)=\frac14\left[x_r+y_r+u_r+v_r-\Phi_{c^M}(x_r-u_r,y_r-v_r)\right].
\]
We use the estimate
\begin{equation}
e_{M,r}\le\frac12\left[\min\{x_r,u_r\}+\min\{y_r,v_r\}\right]+c^M\left[\min\{x_r,v_r\}+\min\{y_r,u_r\}\right].
\label{eq:two-common-pure-branch-bound}
\end{equation}
To prove it, let \(x,y,u,v\ge0\), \(z\in[0,1]\), \(A:=x-u\), and \(B:=y-v\). If \(AB\ge0\), then \(\Phi_z(A,B)=|A+B|\), and the corresponding branch contribution equals either \((u+v)/2\) or \((x+y)/2\), which is exactly the first term on the right-hand side of Eq.~\eqref{eq:two-common-pure-branch-bound}. If \(A\ge0\ge B\), write \(a:=A\) and \(b:=-B\). When \(a+b>0\),
\[
\begin{aligned}
\frac14\left[x+y+u+v-\Phi_z(A,B)\right]&=\frac{u+y}{2}+\frac14\left[a+b-\sqrt{(a+b)^2-4abz}\right]\\
&\le\frac{u+y}{2}+z\frac{ab}{a+b}\le\frac{u+y}{2}+z\min\{x,v\},
\end{aligned}
\]
while the case \(a=b=0\) is immediate. The case \(A\le0\le B\) follows by interchanging \((x,u)\) and \((y,v)\), proving Eq.~\eqref{eq:two-common-pure-branch-bound}.

Define the classical overlap sum
\begin{equation}
L_M(\alpha,\beta):=\sum_{k=0}^M\binom Mk\min\left\{\alpha^k(1-\alpha)^{M-k},\beta^k(1-\beta)^{M-k}\right\}.
\label{eq:classical-Bernoulli-overlap}
\end{equation}
For \(0\le k\le M\), let \(f_k(u,v):=\min\{u^k(1-u)^{M-k},v^k(1-v)^{M-k}\}\). Then
\[
\min\{x_r,u_r\}+\min\{y_r,v_r\}=f_r(p,q)+f_{M-r}(p,q),
\]
and
\[
\min\{x_r,v_r\}+\min\{y_r,u_r\}=f_r(p,1-q)+f_{M-r}(p,1-q).
\]
Moreover, for any \(g_0,\ldots,g_M\),
\[
\sum_{r=0}^{\lfloor M/2\rfloor}\nu_{M,r}(g_r+g_{M-r})=\sum_{k=0}^M\binom Mk g_k.
\]
Indeed, \(r=0\) produces the endpoint terms \(k=0,M\); each \(1\le r<M/2\) pairs \(k=r\) with \(k=M-r\); and, when \(M\) is even, \(\nu_{M,M/2}=\frac12\binom M{M/2}\) compensates for the coincidence \(r=M-r\). Consequently,
\[
\sum_{r=0}^{\lfloor M/2\rfloor}\nu_{M,r}\left[\min\{x_r,u_r\}+\min\{y_r,v_r\}\right]=L_M(p,q),
\]
and
\[
\sum_{r=0}^{\lfloor M/2\rfloor}\nu_{M,r}\left[\min\{x_r,v_r\}+\min\{y_r,u_r\}\right]=L_M(p,1-q).
\]
Since \(P_{e,\mathrm{opt}}^{(M)}=\sum_{r=0}^{\lfloor M/2\rfloor}\nu_{M,r}e_{M,r}\), Eq.~\eqref{eq:two-common-pure-branch-bound} yields
\begin{equation}
P_{e,\mathrm{opt}}^{(M)}\le\frac12L_M(p,q)+c^ML_M(p,1-q).
\label{eq:two-common-pure-error-upper-bound}
\end{equation}

Let \(P_u\) and \(P_v\) denote the Bernoulli distributions with parameters \(u\) and \(v\). Every binary sequence containing \(k\) ones has probabilities \(u^k(1-u)^{M-k}\) and \(v^k(1-v)^{M-k}\) under \(P_u^{\otimes M}\) and \(P_v^{\otimes M}\), respectively, and there are \(\binom Mk\) such sequences. Thus
\[
L_M(u,v)=\sum_{\boldsymbol x\in\{0,1\}^M}\min\{P_u^{\otimes M}(\boldsymbol x),P_v^{\otimes M}(\boldsymbol x)\},
\]
so \(\frac12L_M(u,v)\) is the minimum equal-prior Bayes error for discriminating \(P_u^{\otimes M}\) from \(P_v^{\otimes M}\). The classical Chernoff theorem therefore gives
\[
L_M(u,v)=\exp[-M\xi_{\mathrm B}(u,v)+o(M)].
\]
Hence
\[
\frac12L_M(p,q)=\exp[-M\xi_{\mathrm B}(p,q)+o(M)],
\]
and, for \(c>0\),
\[
c^ML_M(p,1-q)=\exp\left[-M\left(-\log c+\xi_{\mathrm B}(p,1-q)\right)+o(M)\right].
\]
Using \(e^{-Ma+o(M)}+e^{-Mb+o(M)}=e^{-M\min\{a,b\}+o(M)}\) in Eq.~\eqref{eq:two-common-pure-error-upper-bound}, we obtain
\[
P_{e,\mathrm{opt}}^{(M)}\le\exp\left[-M\min\left\{\xi_{\mathrm B}(p,q),-\log c+\xi_{\mathrm B}(p,1-q)\right\}+o(M)\right].
\]
If \(c=0\), the second term in Eq.~\eqref{eq:two-common-pure-error-upper-bound} vanishes for every \(M\ge1\), and the same conclusion holds under the convention \(-\log0=+\infty\). Therefore \(E\ge\xi_{\mathrm{acr}}\). Combining this achievability bound with the fixed-permutation converse gives
\begin{equation}
E=\lim_{M\to\infty}-\frac1M\log P_{e,\mathrm{opt}}^{(M)}=\xi_{\mathrm{acr}}=\min\left\{\xi_{\mathrm B}(p,q),-\log c+\xi_{\mathrm B}(p,1-q)\right\}.
\label{eq:two-common-pure-exact-exponent}
\end{equation}

{\em \textbf{Swapped probabilities.}} A transparent subcase is obtained by setting \(p=a\), \(q=b\), and \(a+b=1\):
\[
H_0:\calE_a=\{(a,\rho_1),(b,\rho_2)\},\qquad H_1:\calE_b=\{(b,\rho_1),(a,\rho_2)\}.
\]
In this case,
\[
B_{r,M-r}=-A_{r,M-r},\qquad A_{r,M-r}=a^rb^{M-r}-b^ra^{M-r},
\]
and hence
\[
\Phi_{c^M}(A_{r,M-r},-A_{r,M-r})=2|A_{r,M-r}|\sqrt{1-c^M}.
\]
Moreover,
\[
\sum_{r=0}^{\lfloor M/2\rfloor}\nu_{M,r}|A_{r,M-r}|=\operatorname{TV}\!\left(\mathrm{Bin}(M,a),\mathrm{Bin}(M,b)\right),
\]
where \(\operatorname{TV}(P,Q):=\frac12\sum_x|P(x)-Q(x)|\). Therefore, the exact equal-prior Bayes error factorizes as
\begin{equation}
P_{e,\mathrm{opt}}^{(M)}=\frac12\left[1-\operatorname{TV}\!\left(\mathrm{Bin}(M,a),\mathrm{Bin}(M,b)\right)\sqrt{1-c^M}\right].
\label{eq:swapped-probability-exact-bayes}
\end{equation}
The corresponding exponent follows from Eq.~\eqref{eq:two-common-pure-exact-exponent}. Since \(\xi_{\mathrm B}(a,1-b)=\xi_{\mathrm B}(a,a)=0\),
\begin{equation}
E=\min\left\{\xi_{\mathrm B}(a,b),-\log c\right\}.
\label{eq:swapped-probability-exponent}
\end{equation}
For \(a=1/3\) and \(b=2/3\), symmetry gives the Bernoulli Chernoff minimizer \(s=1/2\), and
\[
\xi_{\mathrm B}\left(\frac13,\frac23\right)=-\log\left(2\sqrt{\frac13\frac23}\right)=\log\frac{3}{2\sqrt2}.
\]
Thus
\[
E=\min\left\{\log\frac{3}{2\sqrt2},-\log c\right\}.
\]

\section{Mathematical supplements}

\begin{lemma}[Trace norm of a linear combination of two pure states]
\label{lem:trace-norm-two-pure}
Let \(\ket{u}\) and \(\ket{v}\) be normalized vectors and let \(A,B\in\mathbb R\). Then
\begin{equation}
\left\|A\ket{u}\!\bra{u}+B\ket{v}\!\bra{v}\right\|_1=\begin{cases}|A+B|,&AB\ge0,\\[1mm]\sqrt{(A-B)^2+4AB|\langle u|v\rangle|^2},&AB<0.\end{cases}
\label{eq:trace-norm-two-pure}
\end{equation}
\end{lemma}

\begin{proof}
Let \(c:=\langle u|v\rangle\) and \(X:=A\ket{u}\!\bra{u}+B\ket{v}\!\bra{v}\). The support of \(X\) is contained in \(\operatorname{span}\{\ket{u},\ket{v}\}\). For \(|c|<1\), choose the orthonormal basis
\[
\ket{e_1}:=\ket{u},\qquad \ket{e_2}:=\frac{\ket{v}-c\ket{u}}{\sqrt{1-|c|^2}},
\]
so that \(\ket{v}=c\ket{e_1}+\sqrt{1-|c|^2}\ket{e_2}\). In this basis,
\[
X=\begin{pmatrix}A+B|c|^2&Bc\sqrt{1-|c|^2}\\Bc^*\sqrt{1-|c|^2}&B(1-|c|^2)\end{pmatrix}.
\]
Its characteristic polynomial and two possibly nonzero eigenvalues are
\begin{align}
\det(\lambda I-X)&=\lambda^2-(A+B)\lambda+AB(1-|c|^2),\notag\\
\lambda_\pm&=\frac{A+B\pm\sqrt{(A-B)^2+4AB|c|^2}}{2}.
\label{eq:two-pure-linear-combination-eigenvalues}
\end{align}
If \(AB\ge0\), the two terms in \(X\) have the same sign, so \(X\) is positive or negative semidefinite and \(\|X\|_1=|\Tr X|=|A+B|\). If \(AB<0\), then for \(|c|<1\),
\[
\lambda_+\lambda_-=AB(1-|c|^2)<0,
\]
so \(\lambda_+>0>\lambda_-\). Hence
\[
\|X\|_1=|\lambda_+|+|\lambda_-|=\lambda_+-\lambda_-=\sqrt{(A-B)^2+4AB|c|^2}.
\]
Finally, if \(|c|=1\), then \(\ket{v}\!\bra{v}=\ket{u}\!\bra{u}\), so \(X=(A+B)\ket{u}\!\bra{u}\); both cases in Eq.~\eqref{eq:trace-norm-two-pure} reduce to \(\|X\|_1=|A+B|\). This proves the lemma.
\end{proof}

\end{document}